\documentclass[runningheads,envcountsame,orivec]{llncs}

\usepackage{a4wide}
\usepackage{amsfonts, amsmath, amssymb,mathrsfs}
\usepackage{color}
\usepackage[dvipsnames]{xcolor}
\usepackage{tikz}
\usepackage{contour}
\usepackage{bbold}
\usepackage[noadjust]{cite}
\usepackage{bbding}
\usepackage{paralist}

\usepackage{hyperref}
\usepackage[capitalise]{cleveref}
\crefname{section}{Sect.}{Sects.}
\crefname{appendix}{App.}{Apps.}
\crefname{definition}{Def.}{Defs.}
\crefname{proposition}{Prop.}{Props.}
\Crefname{section}{Section}{Sections}
\Crefname{appendix}{Appendix}{Appendices}
\Crefname{definition}{Definition}{Definitions}
\Crefname{proposition}{Proposition}{Propositions}

\newcommand{\eq}{\leftrightarrow}

\newcommand{\imp}{\rightarrow}

\newcommand{\et}{\wedge}
\newcommand{\vel}{\vee}
\newcommand{\Et}{\bigwedge}
\newcommand{\Vel}{\bigvee}

\renewcommand{\phi}{\varphi}

\newcommand{\power}{\mathcal P}

\newcommand{\M}{\widehat{K}}

\newcommand{\lang}{\mathcal L}

\newcommand{\powerset}{\mathcal{P}}

\newcommand{\ourpmb}[1]{\mathbb{#1}}

\usepackage{url}

\newcommand{\HH}{\mathcal H}
\newcommand{\KK}{\mathcal K}
\newcommand{\KH}{\mathit{KH}}

\newcommand{\ce}{\colonequals}

\newcommand{\lanhopekn}{\mathcal{L}_{\mathit{KH}}}

\newcommand{\correct}[1]{\mathit{correct_{#1}}}

\newcommand{\oneH}{\mathrm{one}\HH}
\newcommand{\HinK}{\HH\mathrm{in}\KK}

\newcommand{\Sfive}{\mathsf{S5}}
\newcommand{\KBfour}{\mathsf{KB4}}

\newcommand{\Prop}{\mathsf{Prop}}

\newcommand{\A}{\mathcal{A}}

\newcommand{\MP}{\mathit{MP}}
\newcommand{\Nec}{\mathit{Nec}}

\newcommand{\axKH}{\mathscr{K\!\!H}}

\newcommand{\byzf}{\mathit{Byz}_{\!f}}

\newcommand{\byz}{\mathit{Byz}}

\newcommand{\pub}{\textit{pub}}
\newcommand{\priv}{\textit{priv}}

\newcommand{\Benthem}[1]{}
\newcommand{\Eijck}[1]{}
\newcommand{\Hoek}[1]{}
\newcommand{\Ditmarsch}[1]{}

\usepackage{colonequals}

\begin{document}
\title{A Logic for Repair and State Recovery in Byzantine Fault-tolerant Multi-agent Systems}
\titlerunning{A Logic for Repair and State Recovery}
\author{Hans van Ditmarsch\inst{1}\orcidID{0000-0003-4526-8687} \and \\ Krisztina Fruzsa\inst{2}\orcidID{0000-0002-2013-1003}\thanks{Was a PhD~student in the  FWF doctoral program LogiCS~(W1255) and also supported by the FWF~project DMAC (P32431).} \and \\ Roman Kuznets\inst{2}\orcidID{0000-0001-5894-8724}\thanks{This research was funded in whole or in part by the Austrian Science Fund (FWF) project ByzDEL~[\href{https://doi.org/10.55776/P33600}{10.55776/P33600}].  For open access purposes, the author has applied a CC BY public copyright license to any author accepted manuscript version arising from this submission.}\Envelope \and \\ Ulrich Schmid\inst{2}\orcidID{0000-0001-9831-8583}
}
\authorrunning{H. van Ditmarsch \and K. Fruzsa \and R. Kuznets \and U. Schmid}
\institute{CNRS, University of Toulouse, IRIT, France \\ \email{hansvanditmarsch@gmail.com}\\
  \and Embedded Computing Systems Group, TU Wien, Austria \\ \email{krisztina.fruzsa@tuwien.ac.at}, \email{\{rkuznets,s\}@ecs.tuwien.ac.at}}

\maketitle
\begin{abstract}
We provide novel epistemic logical language and semantics for  modeling and analysis of byzantine fault-tolerant multi-agent systems, with the intent of not only facilitating reasoning about the agents' fault status but also supporting model updates for repair~and state recovery. Besides the standard knowledge modalities, our logic provides  additional agent-specific hope modalities capable of expressing that an agent is not faulty, and also dynamic modalities enabling change to the agents' correctness status. These dynamic modalities are interpreted as model updates that come in three flavors: fully public, more private, and/or involving factual change. Tailored examples demonstrate the utility and flexibility of our logic for modeling a wide range of fault-detection, isolation, and recovery (FDIR) approaches in mission-critical distributed systems. By providing complete axiomatizations for all  variants of our logic, we also create a  foundation for building future verification tools for this important class of fault-tolerant applications. 
\end{abstract}

\section{Introduction and Overview} 
\label{section:introduction}

\paragraph*{State of the art.} A few years ago, the standard epistemic analysis of distributed systems via the runs-and-systems framework \cite{bookof4,HM90,Mos15TARK} was finally extended~\cite{PKS19:TR,KPSF19:FroCos,KPSF19:TARK} to fault-tolerant systems with (fully) \emph{byzantine} agents~\cite{lamport1982byzantine}.\footnote{The term `byzantine' is not always used consistently in the literature. In some instances, agents were called byzantine despite exhibiting only restricted (sometimes even benign~\cite{DworkM90}) types of faults. In those terms, agents we call `byzantine' in this paper would be called `fully byzantine.'} 
Byzantine agents
 constitute the worst-case scenario in terms of fault-tolerance: not only can they arbitrarily deviate from their respective protocols, but the perception of their own actions and observed events can be corrupted, possibly unbeknownst to them, resulting in false memories. 
Whether byzantine agents are actually present in a system, the very possibility of their presence has drastic and debilitating effects on the epistemic state of all agents, including the correct (i.e., non-faulty) ones, due to the inability to rule out so-called \emph{brain-in-a-vat} scenarios~\cite{PutnamBrainBook}: a brain-in-a-vat agent is a faulty agent with completely corrupted perceptions that provide no reliable information about the system~\cite{KPSF19:FroCos}. In such a system, \emph{no}~agent can ever know certain elementary facts, such as their own or some other agent's correctness, no matter whether the system is asynchronous~\cite{KPSF19:FroCos} or synchronous~\cite{schlogl2020persistence}. Agents can, however, sometimes know their own faultiness or obtain belief in some other agents' faultiness~\cite{SS23:TARK}.

In light of knowledge $K_i \phi$ often being unachievable in systems with byzantine agents, \cite{KPSF19:FroCos} also introduced a weaker epistemic notion called~\emph{hope}. It was initially defined as 
\[
H_i \phi \ce \correct{i} \imp K_i(\correct{i} \imp \phi),
\] 
where the designated atom~$\correct{i}$ represents agent~$i$'s correctness. 
In this setting, one can define belief as $B_i \phi \ce K_i(\correct{i} \imp \phi)$ \cite{SS23:TARK}. Hope was successfully used in~\cite{FKS21} to analyze the \emph{Firing Rebels with Relay}~(FRR) problem, which is
the core of the well-known \emph{consistent broadcasting} primitive~\cite{ST87}. Consistent broadcasting has been used as a pivotal building block in fault-tolerant distributed algorithms, e.g.,~for byzantine fault-tolerant clock synchronization~\cite{DFPS14:JCSS,FS12:DC,RS11:TCS,ST87,WS09:DC}, synchronous consensus~\cite{ST87:abc}, and as a general reduction of distributed task solvability in systems with byzantine failures to solvability in systems with crash failures~\cite{MTH14:STOC}.\looseness=-1

The hope modality was first axiomatized in~\cite{Fruzsa23} using $\correct{i}$ as designated atoms. Whereas the resulting logic turned out to be well-suited for modeling and analyzing problems in byzantine fault-tolerant distributed computing systems like FRR~\cite{FKS21}, it is unfortunately not normal. Our long-term goal of also creating the foundations for \emph{automated} verification of such applications hence suggested to look for an alternative axiomatization. In~\cite{hvdetal.aiml:2022}, we presented a normal modal logic that combines $\KBfour_n$ hope modalities with  $\Sfive_n$ knowledge modalities, which is based on defining $\correct{i} \ce \neg H_i \bot$ via frame-characterizable axioms. This logic indeed unlocks powerful techniques developed for normal modal logics both in model checkers like DEMO \cite{Eijck07:dadoem} or MCK~\cite{gammieetal:2004}~and, in particular, in epistemic theorem proving environments such as LWB \cite{heuerdingetal:1996}.

Still, both versions \cite{hvdetal.aiml:2022,Fruzsa23} of the logic of hope target byzantine fault-tolerant distributed systems 
only where, once faulty, agents remain faulty and cannot be ``repaired'' to become correct again.
Indeed, solutions for problems like FRR employ \emph{fault-masking techniques} based on replication~\cite{Sch90}, which prevent the adverse effects of the faulty agents from contaminating the behavior of the correct agents but do not attempt to change the behavior of the faulty agents. Unfortunately, fault masking is only feasible if no more than a certain fraction~$f$ of the overall $n$~agents in the system may become faulty  (e.g., $n \geq 3f+1$ in the case of~FRR). Should it ever happen that more than $f$ agents become faulty in a run, no properties can typically be guaranteed anymore, which would be devastating in mission-critical applications. 

\emph{Fault-detection, isolation, and recovery} (FDIR) is an alternative fault-tol\-er\-ance technique, which attempts to discover and repair agents that became faulty in order to subsequently re-integrate them into the system. The primary target here are permanent faults, which do not go away ``by themselves'' after some time but rather require explicit corrective actions. Pioneering fault-tolerant systems implementations like MAFT \cite{KWFT88} and GUARDS \cite{PABB99} combined fault-masking techniques like byzantine agreement \cite{lamport1982byzantine} and FDIR approaches to harvest the best of both worlds.

Various paradigms have been proposed for implementing the  
steps in FDIR: Fault-detection can be done by a central FDIR unit, which is implemented in some very reliable technology~and oversees the whole distributed system. Alternatively, distributed FDIR employs distributed diagnosis \cite{WLS97}, e.g., based on evidence \cite{AR89}, and is typically combined with byzantine consensus \cite{lamport1982byzantine} to ensure agreement among the replicated FDIR~units. Agents diagnosed as faulty are subsequently forced to reset and execute built-in self tests, possibly followed by repair actions like hardware reconfiguration. Viewed at a very abstract level, the FDI~steps of FDIR thus cause a faulty agent to become correct again.
Becoming correct again is, however, not enough to enable the agent to also participate in the (on-going) execution of the remaining system. The latter also requires a successful \emph{state recovery} step~R, which makes the local state of the agent consistent with the current global system state. Various recovery techniques have been proposed for this purpose, ranging from  pro-active recovery \cite{Rus96}, where the local state of \emph{every} agent is periodically replaced by a majority-voted version, to techniques based on checkpointing \& rollback or message-logging~\&~replay, see \cite{EAWJ02} for a survey. The common aspect of all these techniques is that the local state of the recovering agent is changed based on information originating from other agents.

\paragraph*{Our contribution.} In this paper, we provide the first logic that not only enables one to reason about the fault status of agents, but also provides mechanisms for updating the model so as to change the fault status of agents, as well as their local states. 
Instead of handling such dynamics in the byzantine extension of the runs-and-systems framework ~\cite{PKS19:TR,KPSF19:FroCos,KPSF19:TARK}, i.e., in a temporal epistemic setting, we do it in a dynamic epistemic setting: we restrict our attention to the instants where the ultimate goal of (i)~the~FDI~steps (successfully repairing a faulty processor) and (ii)~the~R~step (recovering the repaired processor's local state) is reached, and investigate the dynamics of the agents' correctness/faultiness and its 
interaction with knowledge at these instants. 

Our approach enables us to separate the issue of (1) verifying the correctness of the specification of an FDIR mechanism from the problem of (2) guaranteeing the correctness of its protocol implementation, and to focus on (1). Indeed, verifying the correctness of the implementation of some specification is the standard problem in formal verification, and powerful tools exist that can be used for this purpose. However, even a fully verified FDIR protocol would be completely useless if the FDIR specification was erroneous from the outset, in the sense that it does not correctly identify and hence repair faulty agents in some cases. Our novel logics and the underlying model update procedures provide, to the best of our knowledge, the first suitable foundations for (1), as~they allow to formally specify (1.a) \emph{when} a model update shall happen, and (1.b) the result of the model update. While we cannot claim that no better approach exists, our various examples at least reveal that we can model many crucial situations arising in FDIR schemes.

In order to introduce the core features of our logic and its update mechanisms, we use a simple example:
Consider two agents $a$ and $b$, each knowing their own local states, where global state $ij$, with $i,j \in \{0,1\}$, means that $a$'s local state is~$i$ and $b$'s local state is~$j$. To describe agent $a$'s local state $i$ we use an atomic proposition $p_a$, where $p_a$~is true if $i=1$ in global state $ij$ and $p_a$ is false if $i=0$, and similarly for $b$'s local state $j$ and atomic proposition $p_b$. 

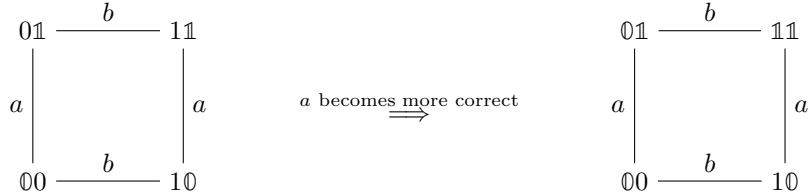
\begin{center}
\begin{tikzpicture}
\node (00) at (0,0) {$\ourpmb{0}0$};
\node (10) at (2,0) {$1\ourpmb{0}$};
\node (01) at (0,2) {$0\ourpmb{1}$};
\node (11) at (2,2) {$1\ourpmb{1}$};
\draw[-] (00) -- node[above] {$b$} (10);
\draw[-] (01) -- node[above] {$b$} (11);
\draw[-] (00) -- node[left] {$a$} (01);
\draw[-] (10) -- node[right] {$a$} (11);
\node (imp) at (5,1) {$\stackrel{\text{$a$ becomes more correct}} \Longrightarrow$};
\node (00r) at (8,0) {$\ourpmb{0}0$};
\node (10r) at (10,0) {$1\ourpmb{0}$};
\node (01r) at (8,2) {$\ourpmb{01}$};
\node (11r) at (10,2) {$\ourpmb{11}$};
\draw[-] (00r) -- node[above] {$b$} (10r);
\draw[-] (01r) -- node[above] {$b$} (11r);
\draw[-] (00r) -- node[left] {$a$} (01r);
\draw[-] (10r) -- node[right] {$a$} (11r);
\end{tikzpicture}
\end{center}

Knowledge and hope of the agents is represented in a Kripke model $M$ for our system consisting of four states (worlds), shown in the left part of the figure above. Knowledge $K_i$ is interpreted by a knowledge relation $\KK_i$ and hope $H_i$ is interpreted by a hope relation $\HH_i$. Worlds that are $\KK_i$-indistinguishable, in the sense that agent $i$ cannot distinguish which of the worlds is the actual one, are connected by an $i$-labeled link, where we assume reflexivity, symmetry, and transitivity. Worlds $ij$ that are in the non-empty part of the $\HH_i$ relation, where agent $i$ is correct, have $i$ outlined as $\ourpmb{0}$ or $\ourpmb{1}$. For example, in the world depicted as $0\ourpmb{1}$ above, agent $a$ is faulty and agent $b$ is correct.\looseness=-1

Now assume that we want agent $a$ to become correct in states $01$ and $11$ where $p_b$ is true. For example, this could be dictated by an FDIR mechanism that caused~$b$ to diagnose $a$ as faulty. Changing the fault status of $a$ accordingly (while not changing the correctness of $b$) results in the updated model on the right in the above figure. Note  that $a$ was correct in state $00$ in the left model, but did not know this, whereas agent $a$ knows that she is correct in state $00$ after the update. Such a model update will be specified in our approach by a suitable \emph{hope update formula} for every agent, which, in the above example, is $\neg H_a \bot \lor p_{b}$ for agent $a$ and $\neg H_{b} \bot$ for agent $b$. Note carefully that every hope update formula implicitly specifies both (a) the situation in the original model in which a change of the hope relation is applied, namely, some agent $i$'s correctness/faultiness status encoded as $\neg H_i\bot/H_i \bot$, and (b) the result of the respective update of the hope relation.

Clearly, different FDIR approaches will require very different hope update formulas for describing their effects. 
In our logic, we provide two basic hope update mechanisms that can be used here: \emph{public} updates, in which  the agents are certain about the exact hope updates occurring at other agents, and \emph{private} updates (strictly speaking, semi-private updates \cite{hvd.jolli:2002}), in which the agents may be uncertain about the particular hope updates occurring at other agents.
The former is suitable for FDIR approaches where a central FDIR unit in the system triggers and coordinates all~FDIR~activities, the latter is needed for some distributed FDIR schemes.

Moreover, whereas the agents' local states do not necessarily have to be changed when becoming correct, FDIR usually
requires to erase traces of erroneous behavior before recovery from the history in the R~step. Our logic hence provides an additional \emph{factual change} mechanism for accomplishing this as well. 
For example, simultaneously with or after becoming correct, agents may also need to change their local state by making false the atomic proposition that records that step 134 of the protocol was (erroneously) executed. 
Analogous to hope update formulas, suitable \emph{factual change formulas} are used to encode when and how atomic propositions will change. 
Besides syntax and semantics, we provide  complete axiomatizations of all variants of our logic, and demonstrate its utility and flexibility for modeling a wide range of FDIR mechanisms by means of many application examples. In order to focus on the essentials, we use only \mbox{2-agent} examples for highlighting particular challenges arising in FDIR. We note, however, that it is usually straightforward to generalize those for more than two agents, and to even combine them for modeling more realistic FDIR scenarios.

\emph{Summary of the utility of our logic.} Besides contributing novel model update mechanisms to the state-of-the-art in dynamic epistemic logic, the main utility of our logic is that it enables epistemic reasoning and verification of FDIR mechanism \emph{specifications}. Indeed, even a fully verified protocol implementation of some FDIR mechanism would be meaningless if its specification allowed unintended effects. Our hope update/factual change formulas
formally and exhaustively specify what the respective model update accomplishes, i.e., encode both the preconditions for changing some agent's fault status/atomic propositions and the actual change. Given an initial model and these update formulas, our logic thus enables one to check (even automatically) whether the updated model has all the properties intended by the designer, whether certain state invariants are preserved by the update, etc. Needless to say, there are many reasons why a chosen specification might be wrong in this respect: the initial model might not provide all the required information,
undesired fault status changes could be triggered in some worlds, or supporting information required for an agent to recover its local state might not be available. The ability to (automatically) verify the absence of such undesired effects of the specification of an FDIR mechanism is hence important in the design of mission-critical distributed systems.

\paragraph*{Paper organization.} \Cref{sec:logic} recalls the syntax and semantics of the logic for knowledge and hope \cite{hvdetal.aiml:2022}. \Cref{sec:public} expands this language with dynamic modalities for publicly changing hope. 
\Cref{sec.semi-private} generalizes the language to private updates. 
In \cref{sec.factual}, we add factual change to our setting. 
Some conclusions in \cref{sec:conclusions} complete our paper.

\section{A Logic of Hope and Knowledge}
\label{sec:logic}

We succinctly present the logic of hope and knowledge \cite{hvdetal.aiml:2022}. Throughout our presentation, let $\A\ce\{1, \dots, n\}$ be a finite set of agents and let $\Prop$ be a non-empty countable set of atomic propositions.

\paragraph*{Syntax.} The language $\lanhopekn$ is  defined as
\begin{equation}
\label{eq:BNF}
\phi \coloncolonequals p \mid \neg \phi \mid (\phi \et \phi) \mid K_i \phi \mid H_i \phi
,
\end{equation}
where $p \in \Prop$ and $i \in \A$. We take $\top$ to be the abbreviation for some fixed propositional tautology and  $\bot$ for $\neg \top$. We also use standard abbreviations for the remaining boolean connectives, $\M_i \phi$ for the dual modality $\neg K_i\neg \phi$ for `agent $a$ considers $\phi$ possible', $\widehat{H}_i \phi$ for $\neg H_i \neg \phi$, and $E_G \phi$ for mutual knowledge $\bigwedge_{i \in G} K_i \phi$ in a group $G \subseteq \A$. Finally, we define belief $B_i \phi$ as $K_i (\neg H_i \bot \imp \phi)$; we recall that $\neg H_i \bot$ means that $i$ is correct.

\paragraph*{Structures.}
A \emph{Kripke model} is a tuple $M=(W,\pi,\KK,\HH)$ where $W$ is a non-empty set of \emph{worlds} (or~\emph{states}), $\pi \colon \Prop \to \powerset(W)$ is a \emph{valuation function} mapping each atomic proposition to the set of worlds where it is true, and $\KK: \A \imp \power(W \times W)$ and $\HH: \A \imp \power(W \times W)$ are functions that assign to each agent $i$ a \emph{knowledge relation} $\KK_i \subseteq W \times W$ respectively a \emph{hope relation} $\HH_i \subseteq W \times W$, where we have written $\KK_i$ resp.\ $\HH_i$ for $\KK(i)$ and $\HH(i)$. We write $\HH_i(w)$ for $\{v \mid (w,v) \in \HH_i\}$ and  $w\HH_iv$ for $(w,v) \in \HH_i$, and similarly for $\KK_i$. We require knowledge relations~$\KK_i$ to be equivalence relations and hope relations $\HH_i$ to be shift-serial (that is, if $w\HH_iv$, then there exists a $z\in W$ such that $v\HH_iz$). In addition, the following conditions should also be satisfied:
\begin{align*}
\HinK&: \qquad \HH_i \subseteq \KK_i,\\
\oneH&: \qquad (\forall w,v \in W) (\HH_i(w)\ne\varnothing \land \HH_i(v)\ne\varnothing \land 
 w\KK_iv  \Longrightarrow w\HH_iv).
\end{align*}
It can be shown that all $\HH_i$ relations are so-called \emph{partial equivalence relations}: they are  transitive and symmetric binary relations \cite{mimo}. 

The class of Kripke models $(W,\pi,\KK,\HH)$ (given~$\A$~and~$\Prop$) is named $\KK\HH$.

\paragraph*{Semantics.} 
We  define truth for formulas~$\phi \in \lanhopekn$ at a world $w$ of a model $M = (W,\pi,\KK,\HH) \in \KK\HH$ in the standard way: in particular,
$M,w \models p$ if{f} $w \in \pi(p)$ where $p \in \Prop$; boolean connectives are classical;
$M,w \models K_i \phi$ if{f}  $M,v \models \phi$  for all~$v$ such that  $w \KK_i v$; and 
$M,w \models H_i \phi$ if{f}  $M,v \models \phi$ for all~$v$ such that  $w \HH_i v$.
A formula $\phi$ is \emph{valid in model~$M$}, denoted~$M\models \phi$, if{f} $M,w \models \phi$ for all~$w \in W$, and it is \emph{valid}, notation $\models \phi$ (or $\KK\HH \models \phi$) if{f} it is valid in all models $M \in \KK\HH$.  \looseness=-1

\paragraph*{Axiomatization.}  The axiom system $\axKH$ for knowledge and hope is given below.
\[\begin{array}{ll@{\quad}|@{\quad}ll}
P & \text{all propositional tautologies} & T^K  & K_i\phi \imp \phi  \\
H^\dagger & H_i \neg H_i \bot &  \KH & H_i \phi \leftrightarrow \bigl(\neg H_i\bot \imp K_i(\neg H_i\bot \imp \phi)\bigr) \\
K^K  & K_{i}(\phi \imp \psi) \land K_{i}\phi \imp K_{i}\psi & \MP & \text{from } \phi \text{ and } \phi \imp \psi, \text{ infer } \psi \\
4^K  & K_{i}\phi \imp K_{i}K_{i}\phi & \Nec^K & \text{from } \phi, \text{ infer } K_i \phi\\
5^K  &\neg K_i\phi \imp K_{i} \neg K_{i}\phi
\end{array}\]

\begin{theorem}[{\cite{hvdetal.aiml:2022}}]
\label{theorem:sckh}
$\axKH$ is sound and complete with respect to $\KK\HH$.
\end{theorem}

\section{Public Hope Update}
\label{sec:public}

\subsection{Syntax and Semantics}

\begin{definition}[Logical language]
Language $\lanhopekn^{\textit{pub}}$ is obtained from $\lanhopekn$ by adding one new  construct: 
\[
\phi \coloncolonequals p \mid \neg \phi \mid (\phi \et \phi) \mid K_i \phi \mid H_i \phi \mid [\underbrace{\phi,\dots,\phi}_n]\varphi
.
\]
\end{definition}
We read a formula of the shape $[\phi_1,\dots,\phi_n]\psi$, often abbreviated as $[\vec{\phi}]\psi$ as follows: after revising or updating hope for agent $i$ with respect to $\phi_i$ for all agents $i \in \A$ simultaneously, $\psi$~(is~true). We call the formula $\phi_i$ the {\em hope update formula for agent $i$}. 

\begin{definition}[Semantics of public hope update]
Let a tuple $\vec{\varphi}\in (\lanhopekn^{\textit{pub}})^n$, a model $M = (W,\pi,\KK, \HH) \in \KK\HH$, and a world $w \in W$ be given. Then 
\[
M, w \models [\vec{\varphi}]\psi \quad \text{if{f}} \quad M^{\vec{\varphi}}, w \models \psi,
\]
where ${M^{\vec{\varphi}}} \ce {(W,\pi,\KK,\HH^{\vec{\varphi}})}$ 
such that for each agent $i \in \mathcal{A}$:
\[w \HH^\chi_i v   \qquad \text{if{f}} \qquad w \KK_i v,   \quad  M,w \models \chi, \quad \text{and} \quad M,v \models \chi\]
and where we write $\HH^{\chi}_i$ for $(\HH^{\vec{\varphi}})_i$ if the $i$-th formula in $\vec{\phi}$ is $\chi$.
\end{definition}
If $M,w \not\models \chi$, then $\HH^{\chi}_{i}(w) = \varnothing$: agent $i$ is faulty in state $w$ after the update, i.e.,~$H_i\bot$~is true. Whereas if $M,w \models \chi$, then $\HH^{\chi}_{i}(w) \neq \varnothing$: agent $i$ is correct in state $w$ after the update, i.e., $\neg H_i\bot$~is true. If the hope update formula for agent~$i$ is $\neg H_i \bot$, then $\neg H_i \bot$ is true in the same states before and after the update. Therefore, $\HH^{\mathstrut\neg H_i \bot}_{i}=\HH_{i}$: the hope relation for $i$ does not change. On the other hand, if the hope update formula for agent~$i$ is  $H_i \bot$, then $\HH^{H_i \bot}_{i}(w) = \varnothing$ if{f} $\HH_i(w) \ne \varnothing$: the~correctness of agent $i$ flips in every state. If we wish to model that agent $i$ becomes \emph{more correct} (in the model), then the hope update formula for agent $i$ should have the shape $\neg H_i \bot \vel \phi$: the~left disjunct~$\neg H_i \bot$ guarantees that in all states where $i$ already was correct, she remains correct.
We write \looseness=-1
\[ [\phi]_i\psi \quad \text{ for } \quad [\neg H_1 \bot, \dots, \neg H_{i-1} \bot, \,\,\phi,\,\, \neg H_{i+1} \bot, \dots, \neg H_n \bot]\psi
\]
Similarly, we write $[\phi]_G\psi$ if the hope update formulas for all agents $i\in G$ is $\phi$ and other agents $j$ have the trivial hope update formula $\neg H_j \bot$.

\begin{proposition}
\label{proposition:updatedmodelstaysinkh}
If $\vec{\varphi}\in (\lanhopekn^{\textit{pub}})^n$ and $M = (W,\pi,\KK,\HH) \in \KK\HH$, then $M^{\vec{\varphi}} \in \KK\HH$.
\end{proposition}

\begin{proof}
Let $i \in \mathcal{A}$ and $\chi$ be the $i$th formula in $\vec{\phi}$. We need to show that  relation~$\HH^{\chi}_i$ is shift-serial and that it satisfies properties $\HinK$ and $\oneH$.
\begin{itemize} 
\item {[shift-serial]:} Let $w \in W$. Assume $v \in \HH^{\chi}_i (w)$, that is, $w \KK_i v$, and $M, w \models \chi$ and $M, v \models \chi$. Now $v \KK_i w $ follows by symmetry of $\KK_i$. Therefore, $\HH^{\chi}_i (v) \ne \varnothing$ since $w \in \HH^{\chi}_i (v)$.
\item {[$\HinK$]:}  This follows by definition.
\item {[$\oneH$]:}  Let $w, v \in W$. Assume that $\HH^{\chi}_i (w) \ne \varnothing$, that $\HH^{\chi}_i (v) \ne \varnothing$, and that $w \KK_i v$. It follows that there exists some $w' \in \HH^{\chi}_i (w)$, implying that $M, w \models \chi$, and $v' \in \HH^{\chi}_i (v)$, implying that $M, v \models \chi$.  Now $w \HH^{\chi}_i v$ follows immediately.\qed
\end{itemize}
\end{proof}

The hope update $\phi$ for an agent $a$ is reminiscent of the refinement semantics of public announcement $\phi$  \cite{jfaketal.jancl:2007}. However, unlike a public announcement, the hope update installs an entirely novel hope relation and discards the old one. 

\subsection{Applications} \label{ssapps}

In this section, we apply the logical semantics just introduced to represent some typical scenarios that occur in FDIR applications. 
We provide several simple two-agent examples. 
\begin{example}[Correction based on agent $b$ having diagnosed $a$ as faulty]
\label{example:bdiagnosesa}
To correct agent $a$ based on $K_{b} H_{a} \bot$, we update agent $a$'s hope relation based on formula $\neg H_a \bot \lor K_{b} H_{a} \bot$ (and agent $b$'s hope relation based on formula $\neg H_{b} \bot$). We recall that the disjunct $\neg H_a \bot$ guarantees that agent~$a$ will stay correct if she already was. The resulting model transformation is:

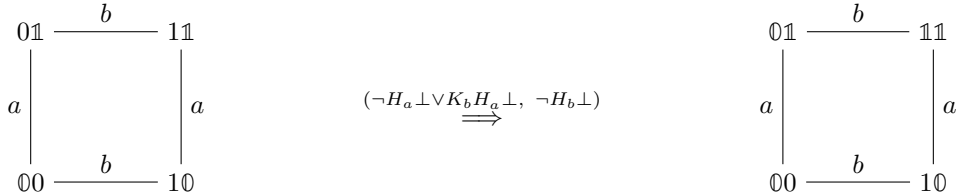
\begin{center}
\begin{tikzpicture}
\node (00) at (0,0) {$\ourpmb{0}0$};
\node (10) at (2,0) {$1\ourpmb{0}$};
\node (01) at (0,2) {$0\ourpmb{1}$};
\node (11) at (2,2) {$1\ourpmb{1}$};
\draw[-] (00) -- node[above] {$b$} (10);
\draw[-] (01) -- node[above] {$b$} (11);
\draw[-] (00) -- node[left] {$a$} (01);
\draw[-] (10) -- node[right] {$a$} (11);
\node (imp) at (6,1) {$\stackrel {(\neg H_a \bot \lor K_{b} H_{a\mathstrut} \bot,\,\,\,\neg H_b\bot)} \Longrightarrow$};
\node (00r) at (10,0) {$\ourpmb{0}0$};
\node (10r) at (12,0) {$1\ourpmb{0}$};
\node (01r) at (10,2) {$\ourpmb{01}$};
\node (11r) at (12,2) {$\ourpmb{11}$};
\draw[-] (00r) -- node[above] {$b$} (10r);
\draw[-] (01r) -- node[above] {$b$} (11r);
\draw[-] (00r) -- node[left] {$a$} (01r);
\draw[-] (10r) -- node[right] {$a$} (11r);
\end{tikzpicture}
\end{center}

\noindent
After the update, in state~$00$, where $a$ was correct but did not know this, and state~$10$, where $a$~knew she was faulty, we get:
\[\begin{array}{l@{\qquad}l}
\text{At state $00$:} 
\\
M,00 \models [\neg H_a \bot \lor K_{b} H_{a} \bot]_a \neg H_{a} \bot & \text{$a$ remains correct}\\
M,00 \models [\neg H_a \bot \lor K_{b} H_{a} \bot]_a K_{a} \neg H_{a} \bot & \text{$a$ learned that she is correct}\\[.5ex]
\text{At state $10$:} 
\\
M,10 \models [\neg H_a \bot \lor K_{b} H_{a} \bot]_a H_{a} \bot & \text{$a$ is still faulty}\\
M,10 \models [\neg H_a \bot \lor K_{b} H_{a} \bot]_a \M_a \neg H_a \bot & \text{$a$ now considers it possible that she is correct}\\
M,10 \models [\neg H_a \bot \lor K_{b} H_{a} \bot]_a K_b \M_a \neg H_a \bot & \text{$b$ learned that}
\\
&\text{\strut\qquad $a$ considers it possible that she is correct}
\end{array}\]
A straightforward generalization of this hope update is correction based on distributed fault detection, where all agents in some sufficiently large group~$G$ need to diagnose agent $a$ as faulty. If~$G$~is fixed, $\neg H_a \bot \lor E_G H_{a} \bot$ achieves this goal. If any group $G$ of at least $k>1$ agents is eligible, then \looseness=-1  
\[
\neg H_a \bot \lor \bigvee_{G\subseteq \A}^{|G|=k} E_G H_{a} \bot
\] 
is the formula of choice.
\end{example}

\begin{example}
\label{example.selfcorrect}
We provide additional examples illustrating the versatility of our approach:
\begin{enumerate}
\item \emph{Self-correction under constraints.}
Unfortunately, \cref{example:bdiagnosesa} cannot be applied in byzantine settings in general, since \emph{knowledge} of other agents' faults is usually not attainable \cite{KPSF19:FroCos}. Hence, one has to either resort to a weaker belief-based alternative or else to an important special case of \cref{example:bdiagnosesa}, namely, \emph{self-correction}, where $G=\{a\}$, i.e., agent $a$ diagnoses itself as faulty. This remains feasible in the byzantine setting because one's own fault is among the few things an agent can know in such systems \cite{KPSF19:FroCos}. Let us illustrate this.

Self-cor\-rect\-ion of agent $a$ without constraints is carried out on the condition that $a$ knows he is faulty ($K_a H_a \bot$). The hope update formula for self-correction of agent $a$ with an optional additional constraint $\phi$ is 
\[ 
\neg H_{a}\bot \lor (\phi \land K_{a} H_a \bot)
\]
where the $\neg H_{a}\bot$ part corresponds to the worlds where agent $a$ is already correct 
and the $\phi \land K_{a} H_a \bot$ part says that, if he knows that he is faulty ($K_a H_a \bot$), then he attempts to self-correct and succeeds if, additionally, a (possibly external) condition $\phi$ holds. Very similarly to \cref{example:bdiagnosesa} we now add an additional constraint $\phi=p_b$. Notice that the update is indeed slightly different than in \cref{example:bdiagnosesa}, as $a$ no longer becomes correct in world $01$.

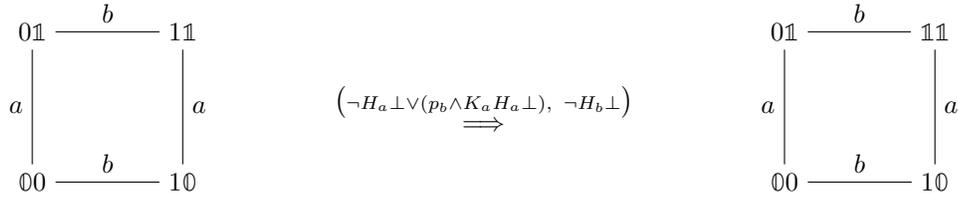
\begin{center}
\begin{tikzpicture}
\node (00) at (0,0) {$\ourpmb{0}0$};
\node (10) at (2,0) {$1\ourpmb{0}$};
\node (01) at (0,2) {$0\ourpmb{1}$};
\node (11) at (2,2) {$1\ourpmb{1}$};
\draw[-] (00) -- node[above] {$b$} (10);
\draw[-] (01) -- node[above] {$b$} (11);
\draw[-] (00) -- node[left] {$a$} (01);
\draw[-] (10) -- node[right] {$a$} (11);
\node (imp) at (6,1) {$\stackrel {\bigl(\neg H_{a}\bot \lor (p_b \land K_{a} H_a \bot),\,\,\, \neg H_b \bot\bigr)} \Longrightarrow$};
\node (00r) at (10,0) {$\ourpmb{0}0$};
\node (10r) at (12,0) {$1\ourpmb{0}$};
\node (01r) at (10,2) {$0\ourpmb{1}$};
\node (11r) at (12,2) {$\ourpmb{11}$};
\draw[-] (00r) -- node[above] {$b$} (10r);
\draw[-] (01r) -- node[above] {$b$} (11r);
\draw[-] (00r) -- node[left] {$a$} (01r);
\draw[-] (10r) -- node[right] {$a$} (11r);
\end{tikzpicture}
\end{center}

After the update, in state~$00$, where $a$ was correct but did not know this, and state~$10$, where $a$~knew she was faulty, we get:
\[\begin{array}{l@{\qquad}l}
\text{At state $00$:} \\
M,00 \models [\neg H_{a}\bot \lor (p_b \land K_{a} H_a \bot )]_a \neg H_{a} \bot & \text{$a$ remains correct}\\
M,00 \models [\neg H_{a}\bot \lor (p_b \land K_{a} H_a \bot)]_a \M_a H_{a} \bot & \text{$a$ still considers it possible  she is faulty}\\[.5ex]
\text{At state $10$:} \\
M,10 \models [\neg H_{a}\bot \lor (p_b \land K_{a} H_a \bot)]_a H_{a} \bot & \text{$a$ remains faulty}\\
 M,10 \models [\neg H_{a}\bot \lor (p_b \land K_{a} H_a \bot)]_a \M_a \neg H_a \bot & \text{$a$ now considers it possible she is correct}\\
 M,10 \models [\neg H_{a}\bot \lor (p_b \land K_{a} H_a \bot)]_a K_b  \M_a \neg H_a \bot & \text{$b$ learned that }
 \\
 & \text{\strut\qquad$a$ considers  it possible she is correct}
\end{array}\]
\item \emph{Update with fail-safe behavior.}
This example specifies a 
variant of self-correction where a faulty agent is only made correct when it knows that it is 
faulty. When it considers it possible that it is correct, however, it deliberately fails itself.  This can be viewed as a way to ensure fail-safe behavior in the case of hazardous system states. What is assumed here is that a deliberately failed agent just stops doing anything, i.e., halts, so that it can subsequently be made correct  via another model update, for example.
In order to specify a model update for fail-safe behavior of agent $a$, the hope update formula $K_a H_a \bot$ can be used. 
The resulting model transformation~is:

\begin{center}
\begin{tikzpicture}
\node (00) at (0,0) {$\ourpmb{0}0$};
\node (10) at (2,0) {$1\ourpmb{0}$};
\node (01) at (0,2) {$0\ourpmb{1}$};
\node (11) at (2,2) {$1\ourpmb{1}$};
\draw[-] (00) -- node[above] {$b$} (10);
\draw[-] (01) -- node[above] {$b$} (11);
\draw[-] (00) -- node[left] {$a$} (01);
\draw[-] (10) -- node[right] {$a$} (11);
\node (imp) at (6,1) {$\stackrel {(K_a H_{a} \bot, \,\,\,\neg H_b \bot)} \Longrightarrow$};
\node (00r) at (10,0) {$00$};
\node (10r) at (12,0) {$\ourpmb{10}$};
\node (01r) at (10,2) {$0\ourpmb{1}$};
\node (11r) at (12,2) {$\ourpmb{11}$};
\draw[-] (00r) -- node[above] {$b$} (10r);
\draw[-] (01r) -- node[above] {$b$} (11r);
\draw[-] (00r) -- node[left] {$a$} (01r);
\draw[-] (10r) -- node[right] {$a$} (11r);
\end{tikzpicture}
\end{center}
After the update, in state~$00$, where $a$ was correct but did not know this, and state~$10$, where $a$~knew she was faulty, we get:
\[\begin{array}{l@{\qquad}l}
\text{At state $00$:} \\
M,00 \models [K_a H_a \bot]_a H_{a} \bot & \text{$a$ became faulty}\\
M,00 \models [K_a H_a \bot]_a K_{a} H_{a} \bot & \text{$a$ learned that she is faulty}\\[.5ex]
\text{At state $10$:} \\
M,10 \models [K_a H_a \bot]_a \neg H_{a} \bot & \text{$a$ became correct}\\
M,10 \models [K_a H_a \bot]_a K_a \neg H_a \bot & \text{$a$ now knows that she is correct}\\
M,10 \models [K_a H_a \bot]_a \M_b K_a \neg H_a \bot & \text{$b$ now  considers it possible}
 \\
 & \text{\strut\qquad$a$ knows she is correct}
\end{array}\]

This hope update would fail agent $a$ also in global states where she \emph{knows} that she is correct, which might seem counterintuitive. In fault-tolerant systems with fully byzantine agents, 
this consideration is moot since agents cannot achieve the knowledge of their own correctness anyway~\cite{KPSF19:FroCos}.

\item \emph{Belief-based correction.}
Since it is generally impossible for agent $b\ne a$ to achieve $K_{b}H_{a}\bot$ in byzantine settings~\cite{KPSF19:FroCos}, correction based on knowledge of faults by other agents is  not implementable in practice. What can, in principle, be achieved in such systems is belief $B_b H_a \bot$ of faults of other agents, where belief is defined as $B_{i} \varphi \ce K_{i} (\neg H_{i} \bot \rightarrow \varphi)$
for any agent $i$ and any formula $\varphi$.

To correct agent $b$ based on agent $a$ believing $b$ to be faulty, we update agent $b$'s hope relation based on formula $\neg H_{b} \bot \lor B_{a} H_{b} \bot$. Note that $B_a H_b \bot$ is indeed initially true in world $00$: if~$a$~is correct, namely only in state $00$ (and not in state $01$), then $b$ is incorrect.

\begin{center}
\begin{tikzpicture}
\node (00) at (0,0) {$\ourpmb{0}0$};
\node (10) at (2,0) {$1\ourpmb{0}$};
\node (01) at (0,2) {$0\ourpmb{1}$};
\node (11) at (2,2) {$1\ourpmb{1}$};
\draw[-] (00) -- node[above] {$b$} (10);
\draw[-] (01) -- node[above] {$b$} (11);
\draw[-] (00) -- node[left] {$a$} (01);
\draw[-] (10) -- node[right] {$a$} (11);
\node (imp) at (6,1) {$\stackrel {(\neg H_a \bot,\,\,\, \neg H_{b} \bot \lor B_{a} H_{b} \bot)} \Longrightarrow$};
\node (00r) at (10,0) {$\ourpmb{00}$};
\node (10r) at (12,0) {$1\ourpmb{0}$};
\node (01r) at (10,2) {$0\ourpmb{1}$};
\node (11r) at (12,2) {$1\ourpmb{1}$};
\draw[-] (00r) -- node[above] {$b$} (10r);
\draw[-] (01r) -- node[above] {$b$} (11r);
\draw[-] (00r) -- node[left] {$a$} (01r);
\draw[-] (10r) -- node[right] {$a$} (11r);
\end{tikzpicture}
\end{center}
After the update, in state~$00$, where $b$ was faulty but did not know this, and state~$10$, where $b$~was correct but did not know this, we get:
\[\begin{array}{l@{\qquad}l}
\text{At state $00$:}\\
M,00 \models [\neg H_{b} \bot \lor B_{a} H_{b} \bot]_b \neg H_{b} \bot & \text{$b$ became correct}\\
M,00 \models [\neg H_{b} \bot \lor B_{a} H_{b} \bot]_b K_{b} \neg H_{b} \bot & \text{$b$ learned that he is correct}\\
M,00 \models [\neg H_{b} \bot \lor B_{a} H_{b} \bot]_b \neg B_{a} H_{b} \bot & \text{$a$ no longer believes that $b$ is faulty}\\[.5ex]
\text{At state $01$:}\\
M,01 \models [\neg H_{b} \bot \lor B_{a} H_{b} \bot]_b K_{a} \neg H_{b} \bot & \text{$a$ learned that $b$ is correct}\\
M,01 \models [\neg H_{b} \bot \lor B_{a} H_{b} \bot]_b K_{b} K_{a} \neg H_{b} \bot & \text{$b$ learned that $a$ knows that $b$ is correct}
\end{array}
\]
Agent $b$ is now correct in all states. Agents $a$ and $b$ therefore have common knowledge that $b$~is correct.

\end{enumerate}
\end{example}

\paragraph*{Byzantine agents.}
We now turn our attention to a different problem that needs to be solved in fault-tolerant distributed systems like MAFT~\cite{KWFT88} and GUARDS~\cite{PABB99} that combine fault-masking approaches with FDIR. What is needed here is to monitor whether there are at most $f$ faulty agents among the $n$ agents in the system, and take countermeasures when the formula 
\[ 
\byzf \ce \Vel_{\substack{G \subseteq \A\\|G|=n-f}} \Et_{i \in G} \neg H_i \bot
\]
is in danger of getting violated or even is violated already.
The most basic way to enforce the global condition $\byzf$ in a hope update is by a constraint on the hope update formulas, rather than by their actual shape. All that is needed here is to ensure, given hope update formulas $\vec{\phi} = (\phi_1,\dots,\phi_n)$, that at least $n-f$ of those are true, which can be expressed by the formula
\[ 
{\vec{\phi}}^{\,n-f} \ce \Vel_{\substack{G \subseteq \A\\|G|=n-f}} \Et_{i \in G} \phi_i.
\] 
We now have the validity 
\[ \models {\vec{\phi}}^{\,n-f} \imp [\vec{\phi}]\byzf. \]
In particular, we also have the weaker 
\[
\models \byzf \et \vec{\phi}^{\,n-f} \imp [\vec{\phi}]\byzf.
\] 
In other words,  
\[
M,w \models \byzf \et \vec{\phi}^{\,n-f} \quad\text{implies}\quad M^{\vec{\phi}},w \models \byzf
.\]
We could also consider generalized schemas such as: $M \models \byzf \et \vec{\phi}^{\,n-f}$ implies $M^{\vec{\phi}} \models \byzf$. In~all these cases, the initial assumption $\byzf$ is superfluous. 

Such a condition is, of course, too abstract for practical purposes. What would be needed here are concrete hope update formulas by which we can update a model when $\byzf$~might become false resp.\ is false already, in which case it must cause the correction of sufficiently many agents to guarantee that $\byzf$ is still true resp.\ becomes true again after the update. Recall that belief~$B_i\psi$ is defined as $K_i(\neg H_i \bot \to \psi)$.
If we define
\[ B_{\geq{f}}\psi \ce \bigvee_{\substack{G\subseteq \A\\|G|= f}} \bigwedge_{i\in G} B_i \psi, \]
it easy to see by the pigeonhole principle that 
\[
\models \byzf \et B_{\geq{f+1}}\psi \imp\psi.
\]
Using 
$\psi = H_a\bot$ will hence result in one fewer faulty agent. To the formula $B_{\geq{f+1}}H_a\bot$ we add a disjunct $\neg H_a \bot$ to ensure correct agents remain correct.
\[ \models \byzf \et B_{\geq{f+1}}H_a\bot \imp [\neg H_a \bot \lor B_{\geq{f+1}}H_a\bot]_a \byz_{\!{f{-}1}}. \]

\subsection{Axiomatization}

Axiomatization $\axKH^{\textit{pub}}$ of the logical semantics for $\lanhopekn^{\textit{pub}}$ extends axiom system~$\axKH$ with axioms describing the interaction between hope updates and other logical connectives. The axiomatization is a straightforward reduction system, where the interesting interaction happens in  hope update binding hope.

\begin{definition}[Axiomatization $\axKH^{\textit{pub}}$]
\label{axiomatization}
$\axKH^{\textit{pub}}$ extends $\axKH$ with axioms
\[\begin{array}{ll@{\quad}l}
{[\vec{\varphi}]} p &\eq p & {[\vec{\varphi}]} K_i \psi \eq K_i [\vec{\varphi}]\psi \\
{[\vec{\varphi}]} \neg\psi &\eq \neg[\vec{\varphi}]\psi & {[\vec{\varphi}]} H_i \psi \eq \bigl(\phi_i \rightarrow K_i(\phi_i \rightarrow [\vec{\varphi}] \psi)\bigr)\\[.3ex]
{[\vec{\varphi}]} (\psi\et\xi) &\eq [\vec{\varphi}]\psi \et [\vec{\varphi}]\xi \quad & {[\vec{\varphi}]} [\vec{\chi}] \psi \eq \bigl[[\vec{\varphi}] \chi_{1}, \dots, [\vec{\varphi}] \chi_{n}\bigr] \psi\\
\end{array}\]
where $\vec{\varphi}=(\varphi_1,\dots,\varphi_n) \in (\lanhopekn^{\textit{pub}})^{n}$, $\vec{\chi}=(\chi_1,\dots,\chi_n) \in (\lanhopekn^{\textit{pub}})^{n}$, $\psi, \xi \in \lanhopekn^{\textit{pub}}$, $p \in \Prop$, and $i \in \mathcal{A}$.
\end{definition}

\begin{theorem}[Soundness]
\label{theorem:soundness}
For all $\phi \in \lanhopekn^{\textit{pub}}$, $
\axKH^{\textit{pub}} \vdash \phi$ implies $\KK\HH \models \phi$.
\end{theorem}
\begin{proof}
In light of \cref{theorem:sckh}, it is sufficient to show the validity of the new axioms. More precisely, we consider an arbitrary model $M = (W,\pi,\KK,\HH) \in \KK\HH$ and state $w \in W$ and show that each axiom is true in state $w$:
\begin{itemize}
\item  
Axiom ${[\vec{\varphi}]} p \eq p$ is valid because
\\
$M, w \models [\vec{\varphi}] p$%
\quad if{f}\quad 
$M^{\vec{\varphi}}, w \models p$%
\quad if{f}\quad 
$w \in \pi(p)$%
\quad if{f}\quad 
$M, w \models p$.
\item 
Axiom  ${[\vec{\varphi}]} \neg\psi \eq \neg[\vec{\varphi}]\psi$ is valid because
\\
$M, w \models [\vec{\varphi}] \neg \psi$%
\quad if{f}\quad 
$M^{\vec{\varphi}}, w \models \neg \psi$%
\quad if{f}\quad 
$M^{\vec{\varphi}}, w \not\models \psi$%
\quad if{f}\quad 
$M, w \not\models [\vec{\varphi}] \psi$%
\quad if{f}\quad 
$M, w \models \neg [\vec{\varphi}] \psi$.
\item 
Axiom  ${[\vec{\varphi}]} (\psi\et\xi) \eq [\vec{\varphi}]\psi \et [\vec{\varphi}]\xi$ is valid because
\\
$M, w \models [\vec{\varphi}] (\psi \land \xi)$%
\quad if{f}\quad 
$M^{\vec{\varphi}}, w \models \psi \land \xi$%
\quad if{f}\quad 
$M^{\vec{\varphi}}, w \models \psi$ and $M^{\vec{\varphi}}, w \models \xi$%
\quad if{f} \\
$M, w \models [\vec{\varphi}] \psi$ and $M, w \models [\vec{\varphi}] \xi$%
\quad if{f}\quad 
$M, w \models [\vec{\varphi}] \psi \land [\vec{\varphi}] \xi$.
\item 
Axiom  ${[\vec{\varphi}]} K_i \psi \eq K_i [\vec{\varphi}]\psi$ is valid because
\\
$M, w \models [\vec{\varphi}] K_i \psi$%
\quad if{f}\quad 
$M^{\vec{\varphi}}, w \models K_i \psi$%
\quad if{f}\quad 
$\bigl(\forall v \in \KK_i (w)\bigr)\, M^{\vec{\varphi}}, v \models \psi$%
\quad if{f} \\ 
$\bigl(\forall v \in \KK_i (w)\bigr)\, M, v \models [\vec{\varphi}] \psi$%
\quad if{f}\quad 
$M, w \models K_i [\vec{\varphi}] \psi$.
\item
Axiom  
${[\vec{\varphi}]} H_i \psi \eq \bigl(\phi_i \rightarrow K_i(\phi_i \rightarrow [\vec{\varphi}] \psi)\bigr)$ is valid because \\
$M, w \models [\vec{\varphi}] H_i \psi$%
\quad if{f}\quad
$M^{\vec{\varphi}}, w \models H_i \psi$%
\quad if{f}\quad
$\bigl(\forall v \in \HH^{\phi_i}_i (w)\bigr)\,\, M^{\vec{\varphi}}, v \models \psi$\quad if{f}\\ 
$(\forall v \in W)
\bigl(v \in \KK_i (w) \,\,\&\,\, M, w \models \phi_i  \,\,\&\,\, M, v \models \phi_i \quad\Longrightarrow\quad M^{\vec{\varphi}}, v \models \psi\bigr)$\quad if{f}\\ 
$M, w \models \phi_i  \quad \Longrightarrow\quad (\forall v \in W)
\bigl(v \in \KK_i (w)   \,\,\&\,\, M, v \models \phi_i \quad\Longrightarrow\quad M^{\vec{\varphi}}, v \models \psi\bigr)$\,\, if{f}\\ 
$M, w \models \phi_i  \quad \Longrightarrow\quad \bigl(\forall v \in \KK_i (w)\bigr)
(   M, v \models \phi_i \quad\Longrightarrow\quad M^{\vec{\varphi}}, v \models \psi)$\,\, if{f}\\ 
$M, w \models \phi_i\quad\Longrightarrow\quad\bigl(\forall v \in \KK_i (w)\bigr)(M, v \models \phi_i\quad\Longrightarrow\quad M, v \models [\vec{\varphi}] \psi)$\quad if{f}\\ 
$M, w \models \phi_i\quad\Longrightarrow\quad\bigl(\forall v \in \KK_i (w)\bigr)\,\,M, v \models \phi_i \rightarrow [\vec{\varphi}] \psi$\quad if{f}\\ 
$M, w \models \phi_i\quad\Longrightarrow\quad M, w \models  K_i(\phi_i \rightarrow [\vec{\varphi}] \psi)$\quad if{f}\\ 
$M, w \models \phi_i \rightarrow K_i(\phi_i \rightarrow [\vec{\varphi}] \psi).$
\item To show the validity of axiom ${[\vec{\varphi}]} [\vec{\chi}] \psi \eq \bigl[[\vec{\varphi}] \chi_{1}, \dots, [\vec{\varphi}] \chi_{n}\bigr] \psi$, we first  show that
\[
(M^{\vec{\varphi}})^{\vec{\chi}}=M^{\bigl([\vec{\varphi}] \chi_{1}, \dots, [\vec{\varphi}] \chi_{n}\bigr)}.
\] 
Since domain $W$, valuation $\pi$, and accessibility relations $\KK_{i}$ for all $i \in \mathcal{A}$ are the same in the initial model $M$ and all updated models, we  only need to show that every agent $i$'s hope accessibility relation~$\bigl((\HH^{\vec{\phi}})^{\vec{\chi}}\bigr)_i= (\HH^{\vec{\phi}})^{\chi_i}_i$   from model $(M^{\vec{\varphi}})^{\vec{\chi}}$ coincides  with $i$'s~hope accessibility relation~$\HH^{[\vec{\varphi}] \chi_{i}}_{i}$ from model $M^{\bigl([\vec{\varphi}] \chi_{1}, \dots, [\vec{\varphi}] \chi_{n}\bigr)}$:\\
$w(\HH^{\vec{\phi}})^{\chi_i}_i v$%
\qquad\quad if{f}\qquad\quad 
$w\KK_{i}v$, and $M^{\vec{\varphi}}, w \models \chi_{i}$, and $M^{\vec{\varphi}}, v \models \chi_{i}$%
\qquad\quad if{f}\\  
\strut\hfill
$w\KK_{i}v$, and $M, w \models [\vec{\varphi}] \chi_{i}$, and $M, v \models [\vec{\varphi}] \chi_{i}$%
\qquad\quad if{f}\qquad\quad 
$w \HH^{[\vec{\varphi}]\chi_{i}}_{i} v$.\\ 
It remains to note that $M, w \models  [\vec{\varphi}] [\vec{\chi}] \psi$%
\quad if{f}\quad
$M^{\vec{\varphi}}, w \models [\vec{\chi}] \psi$%
\quad if{f}\quad
$(M^{\vec{\varphi}})^{\vec{\chi}}, w \models \psi$%
\quad if{f}\\ 
$M^{\bigl([\vec{\varphi}] \chi_{1}, \dots, [\vec{\varphi}] \chi_{n}\bigr)}, w \models \psi$%
\quad if{f}\quad 
$M, w \models \bigl[[\vec{\varphi}] \chi_{1}, \dots, [\vec{\varphi}] \chi_{n}\bigr] \psi$.
\qed
\end{itemize}

\end{proof}

Every formula in $\lanhopekn^{\textit{pub}}$ is provably equivalent to a formula in $\lanhopekn$ (\cref{lem:trans}).  
To prove this, we first define the \emph{weight} or \emph{complexity} of a given formula (\cref{definition:weight}) and show a number of inequalities comparing the left-hand side to the right-hand side of the reduction axioms in axiomatization $\axKH^{\textit{pub}}$ (\cref{unequal}). Subsequently, we define a translation from $\lanhopekn^{\textit{pub}}$ to $\lanhopekn$ (\cref{translation}) and finally  show that the translation is a terminating rewrite procedure (\cref{termination}).

\begin{definition}
\label{definition:weight}
The \emph{complexity} $c: \lanhopekn^{\textit{pub}} \to \mathbb{N}$ of  $\lanhopekn^{pub}$-formulas is defined recursively, where $p \in \Prop$, $i \in \mathcal{A}$, and $c(\vec{\phi}) \ce \max \{ c(\phi_i) \mid 1 \leq i \leq n \}$:  
\[\begin{array}{lll}
c(p) &\ce 1 & c(K_{i} \phi) \ce c(\phi) + 1\\
c(\neg \phi) &\ce c(\phi) + 1 & c(H_{i} \phi) \ce c(\phi) + 4\\
c(\phi \land \xi) &\ce \max\{c(\phi),c(\xi)\} + 1 \quad & c\bigl([\vec{\phi}] \xi\bigr) \ce \bigl(c(\vec{\phi}) + 1\bigr) \cdot c(\xi)
\end{array}\]
\end{definition}

\begin{lemma} 
\label{unequal}
For each axiom $\theta_l \eq \theta_r$ from \cref{axiomatization}, $c(\theta_l) > c(\theta_r)$.
\end{lemma}
\begin{proof}
\begin{itemize}
\item 
For  axiom $ {[\vec{\varphi}]} p \eq p$:
\[
c\bigl([\vec{\varphi}] p\bigr) =\bigl(c(\vec{\varphi}) + 1\bigr) \cdot c(p) 
>c(p).
\]
\item 
For axiom ${[\vec{\varphi}]} \neg\psi \eq \neg[\vec{\varphi}]\psi$:
\begin{multline*}
c\bigl([\vec{\varphi}] \neg \psi\bigr) =\bigl(c(\vec{\varphi}) + 1\bigr) \cdot c(\neg \psi) 
= \bigl(c(\vec{\varphi}) + 1\bigr) \cdot \bigl(c(\psi) + 1\bigr) \\
>\bigl(c(\vec{\varphi}) + 1\bigr) \cdot c(\psi) + 1 
=c\bigl([\vec{\varphi}] \psi\bigr) + 1 
=c\bigl(\neg [\vec{\varphi}] \psi\bigr).
\end{multline*}
\item 
For axiom ${[\vec{\varphi}]} (\psi\et\xi) \eq [\vec{\varphi}]\psi \et [\vec{\varphi}]\xi$:
\begin{multline*}
c\bigl([\vec{\varphi}] (\psi \land \xi)\bigr) =\bigl(c(\vec{\varphi}) + 1\bigr) \cdot c(\psi \land \xi) 
=\bigl(c(\vec{\varphi}) + 1\bigr) \cdot \bigl(\max\{c(\psi),c(\xi)\} + 1\bigr) \\
>\bigl(c(\vec{\varphi}) + 1\bigr) \cdot \max\{c(\psi),c(\xi)\} + 1 
=\max\Bigl\{\bigl(c(\vec{\varphi}) +1\bigr) \cdot c(\psi),\,\,\bigl(c(\vec{\varphi}) +1\bigr) \cdot c(\xi)\Bigr\} + 1 \\
=\max\Bigl\{c\bigl([\vec{\varphi}] \psi\bigr),\,\,c\bigl([\vec{\varphi}] \xi\bigr)\Bigr\} + 1 
=c\bigl([\vec{\varphi}] \psi \land [\vec{\varphi}] \xi\bigr).
\end{multline*}
\item
For axiom ${[\vec{\varphi}]} K_i \psi \eq K_i [\vec{\varphi}]\psi$:
\begin{multline*}
c\bigl([\vec{\varphi}] K_{i}\psi\bigr) =\bigl(c(\vec{\varphi}) + 1\bigr) \cdot c(K_{i}\psi) 
= \bigl(c(\vec{\varphi}) + 1\bigr) \cdot \bigl(c(\psi) + 1\bigr) \\
>\bigl(c(\vec{\varphi}) + 1\bigr) \cdot c(\psi) + 1 
=c\bigl([\vec{\varphi}] \psi\bigr) + 1 
=c\bigl(K_{i}[\vec{\varphi}] \psi\bigr).
\end{multline*}
\item
For axiom ${[\vec{\varphi}]} H_i \psi \eq \bigl(\phi_i \rightarrow K_i(\phi_i \rightarrow [\vec{\varphi}] \psi)\bigr)$, given  that 
\[
\phi_i \rightarrow K_i(\phi_i \rightarrow [\vec{\varphi}] \psi) 
\quad=\quad 
\neg \bigl(\phi_i \et \neg K_i\neg(\phi_i \et \neg [\vec{\varphi}] \psi)\bigr):
\]
\begin{multline*}
c\bigl([\vec{\varphi}] H_{i}\psi\bigr) =\bigl(c(\vec{\varphi}) + 1\bigr) \cdot c(H_{i}\psi) 
=\bigl(c(\vec{\varphi}) + 1\bigr) \cdot \bigl(c(\psi) + 4\bigr) 
=\bigl(c(\vec{\varphi}) + 1\bigr) \cdot c(\psi) + \bigl(c(\vec{\varphi}) + 1\bigr) \cdot 4 \\
=c\bigl([\vec{\varphi}] \psi\bigr) + 4 \cdot c(\vec{\varphi}) + 4 
>c\bigl([\vec{\varphi}] \psi\bigr)+ 7 
=\max\Bigl\{c(\varphi_{i}),\,\, c\bigl([\vec{\varphi}] \psi\bigr)+ 5\Bigr\} + 2 \\
=\max\Bigl\{c(\varphi_{i}),\quad\max\bigl\{c(\varphi_{i}),\,c\bigl([\vec{\varphi}] \psi\bigr)+ 1\bigr\} + 4\Bigr\} + 2 \\
=\max\Bigl\{c(\varphi_{i}),\quad\max\bigl\{c(\varphi_{i}),\,c\bigl(\neg [\vec{\varphi}] \psi\bigr)\bigr\} + 4\Bigr\} + 2 \\
=\max\Bigl\{c(\varphi_{i}),\,\, c\bigl(\varphi_{i} \land \neg [\vec{\varphi}] \psi\bigr) + 3\Bigr\} + 2 \\
=\max\Bigl\{c(\varphi_{i}),\,\,c\bigl(\neg (\varphi_{i} \land \neg [\vec{\varphi}] \psi)\bigr) + 2\Bigr\} + 2 \\
=\max\Bigl\{c(\varphi_{i}),\,\,c\bigl(K_{i}\neg (\varphi_{i} \land \neg [\vec{\varphi}] \psi)\bigr) + 1\Bigr\} + 2 \\
=\max\Bigl\{c(\varphi_{i}),\,\,c\bigl(\neg K_{i}\neg (\varphi_{i} \land \neg [\vec{\varphi}] \psi)\bigr)\Bigr\} + 2 \\
=c\bigl(\varphi_{i} \land \neg K_{i}\neg (\varphi_{i} \land \neg [\vec{\varphi}] \psi)\bigr) + 1 
=c\Bigl(\neg \bigl(\varphi_{i} \land \neg K_{i}\neg (\varphi_{i} \land \neg [\vec{\varphi}] \psi)\bigr)\Bigr).
\end{multline*}
\item For axiom ${[\vec{\varphi}]} [\vec{\chi}] \psi \eq \bigl[[\vec{\varphi}] \chi_{1}, \dots, [\vec{\varphi}] \chi_{n}\bigr] \psi$:
\begin{multline*}
c\bigl([\vec{\varphi}] [\vec{\chi}] \psi\bigr) =\bigl(c(\vec{\varphi}) + 1\bigr) \cdot c\bigl([\vec{\chi}] \psi\bigr) 
=\bigl(c(\vec{\varphi}) + 1\bigr) \cdot \bigl(c(\vec{\chi}) + 1\bigr) \cdot c(\psi) 
>\Bigl(\bigl(c(\vec{\varphi}) + 1\bigr) \cdot c(\vec{\chi}) + 1 \Bigr) \cdot c(\psi) \\
= \Bigl(\max\bigl\{\bigl(c(\vec{\varphi}) + 1\bigr) \cdot c(\chi_{1}), \dots, \bigl(c(\vec{\varphi}) + 1\bigr) \cdot c(\chi_{n})\bigr\} + 1\Bigr) \cdot c(\psi) \\
=\Bigl(\max\bigl\{c\bigl([\vec{\varphi}] \chi_{1}\bigr), \dots, c\bigl([\vec{\varphi}] \chi_{n}\bigr)\bigr\} + 1\Bigr) \cdot c(\psi) 
=c\Bigl(\bigl[[\vec{\varphi}] \chi_{1}, \dots, [\vec{\varphi}] \chi_{n}\bigr] \psi\Bigr).\quad
\qed
\end{multline*}
\end{itemize}
\end{proof}

\begin{definition}
\label{translation}
The \emph{translation} $t: \lanhopekn^{pub} \to \lanhopekn$  is defined recursively, where $p \in \Prop$, $i \in \mathcal{A}$, and the $i$th formula of $\vec{\phi}$ is $\phi_i$:
\[
\begin{array}{llll}
t(p) &\ce p & t\bigl([\vec{\phi}]p\bigr) &\ce p \\[.3ex]
t(\neg \phi) &\ce \neg t(\phi) & t\bigl([\vec{\phi}]\neg\xi\bigr) &\ce \neg t\bigl([\vec{\phi}]\xi\bigr) \\[.3ex]
t(\phi \land \xi) &\ce t(\phi) \land t(\xi) \quad & t\bigl([\vec{\phi}] (\xi\et\chi)\bigr) &\ce t\bigl([\vec{\phi}]\xi \et [\vec{\phi}]\chi\bigr) \\[.3ex]
t(K_{i} \phi) &\ce K_i t(\phi) & t\bigl([\vec{\phi}] K_i \xi\bigr) &\ce t\bigl(K_i [\vec{\phi}]\xi\bigr) \\[.3ex]
t(H_{i} \phi) &\ce H_i t(\phi) & t\bigl([\vec{\phi}] H_i \xi\bigr) &\ce t\bigl(\phi_i \rightarrow K_i(\phi_i \rightarrow [\vec{\phi}] \xi)\bigr) \\[.3ex]
& & t\bigl([\vec{\phi}] [\chi_{1}, \dots, \chi_{n}] \xi\bigr) &\ce t\bigl(\bigl[[\vec{\phi}] \chi_{1}, \dots, [\vec{\phi}] \chi_{n}\bigr] \xi\bigr)
\end{array}
\]
\end{definition}

\begin{proposition}[Termination]
\label{termination}
For all $\phi \in \lanhopekn^{\textit{pub}}$, $t(\phi) \in \lanhopekn$.
\end{proposition}
\begin{proof}
This follows by  induction on $c(\varphi)$.\qed
\end{proof}

\begin{lemma}[Equiexpressivity] 
\label{lem:trans}
Language $\lanhopekn^{\textit{pub}}$ is equiexpressive with $\lanhopekn$.
\end{lemma}
\begin{proof}
It follows 
by  induction on $c(\varphi)$ that  $\axKH^{\textit{pub}} \vdash \phi \leftrightarrow t(\phi)$ for all $\phi \in \lanhopekn^{\textit{pub}}$, where, by \cref{termination}, $t(\phi) \in \lanhopekn$.\qed
\end{proof}

\begin{theorem}[Soundness and completeness]
\label{th:pub_complete}
For all $\phi \in \lanhopekn^{\textit{pub}}$, 
\[
\axKH^{\textit{pub}} \vdash \phi\qquad \Longleftrightarrow\qquad\KK\HH \models \phi.
\]
\end{theorem}
\begin{proof}
Soundness was proved in \cref{theorem:soundness}. To prove completeness,
assume $\KK\HH \models \phi$. According to \cref{lem:trans}, we have $\axKH^{\textit{pub}} \vdash \phi \leftrightarrow t(\phi)$. Therefore, by \cref{theorem:soundness}, $\KK\HH \models \phi \leftrightarrow t(\phi)$ follows. Since $\KK\HH \models \phi$ (by assumption), we obtain $\KK\HH \models t(\phi)$. By applying \cref{theorem:sckh}, $\axKH \vdash t(\phi)$ further follows. Consequently, $\axKH^{\textit{pub}} \vdash t(\phi)$. Finally, since $\axKH^{\textit{pub}} \vdash \phi \leftrightarrow t(\phi)$, $\axKH^{\textit{pub}} \vdash \phi$.%
\qed
\end{proof}
\begin{corollary}[Necessitation for public hope updates]
\label{cor:nec}
For all $\psi \in \lanhopekn^{\textit{pub}}$ and $\vec{\phi} \in (\lanhopekn^{\textit{pub}})^n$, 
\[
\axKH^{\textit{pub}} \vdash \psi
\qquad \Longrightarrow\qquad
\axKH^{\textit{pub}} \vdash [\vec{\phi}] \psi.
\]
\end{corollary}
\begin{proof}
Assume $\axKH^{\textit{pub}} \vdash \psi$. By \cref{th:pub_complete}, $\KK\HH \models\psi$. In particular, for any $M=(W,\pi,\KK,\HH) \in \KK\HH$, we have $M^{\vec{\phi}} \models \psi$ since $M^{\vec{\phi}} \in \KK\HH$ by \cref{proposition:updatedmodelstaysinkh}. Thus, $M^{\vec{\phi}}, w \models \psi$  for all $w \in W$. In other words, $M, w \models [\vec{\phi}]\psi$ for all $w \in W$, i.e., $M \models[\vec{\phi}]\psi$. Since $\KK\HH \models[\vec{\phi}]\psi$, we get $\axKH^{\textit{pub}} \vdash [\vec{\phi}] \psi$ by \cref{th:pub_complete}.\qed
\end{proof}

\section{Private Hope Update} 
\label{sec.semi-private}

In the case of the public hope update mechanism introduced in \cref{sec:public},  after the update there is no uncertainty about what happened. In some distributed FDIR schemes, including self-correction, however, the hope update at an agent occurs in a less public way. To increase the application coverage of our logic, we therefore provide the alternative of private hope updates. For that, we use structures inspired by action models. Strictly speaking, such updates are known as \emph{semi-private} (or \emph{semi-public}) updates, as the agents are aware of their uncertainty and know what they are uncertain about, whereas in fully private update the agent does not know that the action took place \cite{hvd.jolli:2002} and may, in fact, believe that nothing happened. The resulting language can be viewed as a generalization of~$\lanhopekn^{pub}$, where the latter now becomes a special case. 
 
\subsection{Syntax and Semantics}

\begin{definition}[Hope update model]\label{def.x}
A \emph{hope update model} for a logical language $\lang$ is a tuple 
\[
U = (E,\vartheta,\KK^{U})
\]
where $E$ is a finite non-empty set of \emph{actions}, $\vartheta : E \to (\A \to \lang)$ is a \emph{hope update function}, and $\KK^U: \A \to \power(E \times E)$ such that all $\KK^{U}_i$~are equivalence relations. For $\vartheta(e)(i)$ we write~$\vartheta_i(e)$. As before, formulas~$\vartheta_i(e)\in\lang$ are \emph{hope update formulas}. 
A \emph{pointed hope update model} (for the logical language $\lang$) is a pair $(U,e)$ where $e \in E$.
\end{definition}

\begin{definition}[Language $\lanhopekn^{priv}$] \label{def.y}
Language $\lanhopekn^{priv}$ is obtained from $\lanhopekn$ by adding one new  construct: 
\[
\phi \coloncolonequals p \mid \neg \phi \mid (\phi \et \phi) \mid K_i \phi \mid H_i \phi \mid [U,e] \phi
\]
where $(U,e)$ is a pointed hope update model for language~$\lanhopekn^{priv}$.
\end{definition}
\Cref{def.y} is given by mutual recursion as usual:  formulas may include hope update models while  hope update models  must include formulas to be used as hope update formulas. All (pointed) hope update models till the end of this section are for language~$\lanhopekn^{priv}$.

\begin{definition}[Semantics of private hope update] \label{def.semanticssemi}
Let $U = (E,\vartheta,\KK^U)$ be a hope update model, $M = (W,\pi,\KK,\HH) \in \KK\HH$, $w \in W$, and $e \in E$. Then:
\[
M,w \models [U,e]\phi \quad \text{if{f}} \quad M \times U,  (w,e) \models \phi,
\]
where $M \times U = (W^\times,\pi^\times,\KK^\times, \HH^\times)$ is such that:
\[\begin{array}{l@{\quad}l@{\quad}l}
W^\times & \ce & W \times E \\
(w,e) \in \pi^\times(p) & \text{if{f}} & w \in \pi(p) \\
(w,e) \KK^\times_i (v,f) & \text{if{f}} & w \KK_i v \text{ and } e \KK^{U}_i f \\
(w,e) \HH^\times_i (v,f) & \text{if{f}} & (w,e) \KK^\times_i (v,f), \text{ and }  M,w \models \vartheta_i(e), \text{ and } M,v \models \vartheta_i(f)
\end{array}\]
\end{definition}
Public hope updates can be viewed as singleton hope update models. Given formulas $\vec{\varphi} \in (\lanhopekn^{pub})^n$, define $ \pub \ce (\{e\}, \vartheta, \KK^\pub)$, where $\vartheta_i(e)\ce\varphi_i$ and $\KK^\pub\ce \{(e,e)\}$.

\paragraph*{Difference with action models.}
Although our hope update models look like action models, they are not really action models in the sense of \cite{baltagetal:1998}. Our actions do not have executability preconditions, such that the updated model is not a restricted modal product but rather the full product. Another difference is that, by analogy with Kripke models for knowledge and hope, we would then have expected a hope relation in the update models. But there is none in our approach.

\begin{proposition}
\label{proposition:semi-privateupdatestayinkh}
$M \times U \in \KK\HH$ for any   hope update model  $U$ and  $M\in \KK\HH$.
\end{proposition}
\begin{proof}
The proof is somewhat similar to that of \cref{proposition:updatedmodelstaysinkh}. It is obvious that all~$\KK^\times_i$ are equivalence relations. Let us show now that for all $i \in \mathcal{A}$  relations $\HH^{\times}_i$ are shift-serial and that they satisfy the properties $\HinK$ and $\oneH$.
\begin{itemize}
\item {\bf $\HH^\times_i$ is shift-serial:} Let $(w,e) \in W^\times$. Assume $(w,e) \HH^\times_i (v,f)$, that is,\linebreak $(w,e) \KK^\times_i (v,f)$, and $M,w \models \vartheta_i(e)$, and $M,v \models \vartheta_i(f)$. $(v,f) \KK^\times_i (w,e)$ follows by symmetry of~$\KK^\times_i$. Therefore, $\HH^\times_i \bigl((v,f)\bigr) \ne \varnothing$ since $(w,e) \in \HH^\times_i \bigl((v,f)\bigr)$.
\item {\bf $\HH^\times_i$ satisfies $\HinK$}:  This follows by definition.
\item {\bf $\HH^\times_i$ satisfies $\oneH$}:  Let $(w,e), (v,f) \in W^\times$. Assume that $\HH^\times_i \bigl((w,e)\bigr) \ne \varnothing$, $\HH^\times_i \bigl((v,f)\bigr) \ne \varnothing$, and $(w,e) \KK^\times_i (v,f)$. As $\HH^\times_i \bigl((w,e)\bigr) \ne \varnothing$, \mbox{$M,w \models \vartheta_i(e)$}. As $\HH^\times_i \bigl((v,f)\bigr) \ne \varnothing$, $M,v \models \vartheta_i(f)$. Therefore, $(w,e) \HH^\times_i (v,f)$.\qed
\end{itemize}
\end{proof}

\begin{definition}
\label{def:comp}
Let $U = (E,\vartheta,\KK^{U})$ and $U' = (E',\vartheta',\KK^{U'})$ be hope update models. The \emph{composition} $(U ; U')$ is $(E'',\vartheta'',\KK^{U ; U'})$ such that:
\[\begin{array}{l@{\quad}l@{\quad}l}
E'' & \ce & E \times E' \\
\smallskip
\vartheta''_i\bigl((e,e')\bigr) & \ce & [U,e] \vartheta_i'(e')\\
\smallskip
(e,e') \KK^{U ; U'}_i (f,f') & \text{if{f}} & e \KK^{U}_i f \text{ and } e' \KK^{U'}_i f'
\end{array}\]
\end{definition}
Since $\KK^{U}_i$ and $\KK^{U'}_i$ are equivalence relations, $\KK^{U ; U'}_i$ is also an equivalence relation, so that $(U ; U')$~is a hope update model.

\subsection{Applications}

The arguably most important usage of private updates in distributed FDIR 
is to express the uncertainty of agents about whether an update affects other agents. 
\looseness=-1

\begin{example}
\label{example:semiprivate}
We present several uses of private hope updates:

\begin{enumerate}
\item \emph{Private correction.}
We reconsider the example from \cref{section:introduction}, only this time we privately correct agent $a$ based on $p_b$ such that agent~$b$ is uncertain whether the hope update happens. This can be modeled by two hope update formulas for agent $a$: $\neg H_a \bot \vel p_b$ and $\neg H_a \bot$. With $\neg H_a \bot \vel p_b$ we associate an event~$c_{p_b}$ where the correction takes place based on the additional constraint~$p_b$, and with $\neg H_a \bot$ we associate an event $noc$ where correction does not take place. Writing $\vartheta(e) = \bigl((\vartheta_a(e),\vartheta_b(e)\bigr)$,  we get $U\ce(E, \vartheta, \KK^{U})$, where:
\[\begin{array}{ll@{\quad}l}
E &\ce \{c_{p_b}, noc\} & \KK^{U}_{a} \ce \text{the identity relation } \{(e,e) \mid e \in E\}\\
\vartheta (c_{p_b}) &\ce (\neg H_a \bot \vel p_b,\neg H_b \bot) \quad & \KK^{U}_{b} \ce \text{the universal relation } E \times E\\
\vartheta (noc) &\ce (\neg H_a \bot,\neg H_b \bot)
\end{array}\]

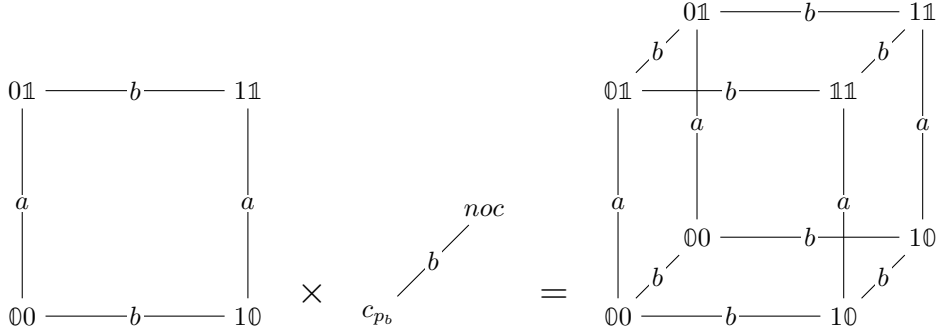
\[
\begin{tikzpicture}[z=0.35cm]
\node (000) at (0,0,0) {$\ourpmb{0}0$};
\node (010) at (0,3,0) {$0\ourpmb{1}$};
\node (100) at (3,0,0) {$1\ourpmb{0}$};
\node (110) at (3,3,0) {$1\ourpmb{1}$};
\draw (000) -- node[fill=white,inner sep=1pt] {$b$} (100);
\draw (000) -- node[fill=white,inner sep=1pt] {$a$} (010);
\draw (010) -- node[fill=white,inner sep=1pt] {$b$} (110);
\draw (100) -- node[fill=white,inner sep=1pt] {$a$} (110);
\end{tikzpicture}
\quad \raisebox{.4cm}{\Large $\times$} \quad 
\begin{tikzpicture}[z=0.35cm]
\node (000) at (0,0,0) {$c_{p_b}$};
\node (010) at (0,0,4) {$noc$};
\draw (000) -- node[fill=white,inner sep=1pt] {$b$} (010);
\end{tikzpicture}
\quad \raisebox{.4cm}{\Large $=$} \quad 
\begin{tikzpicture}[z=0.35cm]
\node (000) at (0,0,0) {$\ourpmb{0}0$};
\node (001) at (0,0,3) {$\ourpmb{0}0$};
\node (010) at (0,3,0) {$\ourpmb{01}$};
\node (011) at (0,3,3) {$0\ourpmb{1}$};
\node (100) at (3,0,0) {$1\ourpmb{0}$};
\node (101) at (3,0,3) {$1\ourpmb{0}$};
\node (110) at (3,3,0) {$\ourpmb{11}$};
\node (111) at (3,3,3) {$1\ourpmb{1}$};
\draw (000) -- node[fill=white,inner sep=1pt] {$b$} (100);
\draw (001) -- node[fill=white,inner sep=1pt] {$b$} (101);
\draw (011) -- node[fill=white,inner sep=1pt] {$b$} (111);

\draw (000) -- node[fill=white,inner sep=1pt] {$a$} (010);
\draw (001) -- node[fill=white,inner sep=1pt] {$a$} (011);
\draw (101) -- node[fill=white,inner sep=1pt] {$a$} (111);

\draw (000) -- node[fill=white,inner sep=1pt] {$b$} (001);
\draw (010) -- node[fill=white,inner sep=1pt] {$b$} (011);
\draw (100) -- node[fill=white,inner sep=1pt] {$b$} (101);
\draw (110) -- node[fill=white,inner sep=1pt] {$b$} (111);
\draw (010) -- node[fill=white,inner sep=1pt] {$b$} (110);
\draw (100) -- node[fill=white,inner sep=1pt] {$a$} (110);
\end{tikzpicture}
\]
When labeling worlds in the figure above, we have abstracted away from the event being executed in a world. Having the same name, therefore, does not mean being the same world. For example, the world $\ourpmb{01}$ at the front of the cube `really' is the pair $(01,c_{p_b})$ with \mbox{$H_a\bigl((01,c_{p_b})\bigr) \ne \varnothing$} and $H_b\bigl((01,c_{p_b})\bigr) \ne \varnothing$. 
We now have for example that, in state $01$, where $b$ knew that $a$ was faulty but $a$ herself did not know  this:

\[\begin{array}{l@{\quad}l}
M, 01 \models  [U,c_{p_b}] (\neg H_{a} \bot \et K_a \neg H_a \bot) & \text{$a$ became correct and now knows she is correct}\\
M, 01 \models [U,c_{p_b}] \neg K_{b} K_{a} \neg H_{a} \bot & \text{$b$ does not know that $a$ knows she is correct}\\
M, 01 \models [U,c_{p_b}] \neg (K_{b} H_{a} \bot \vel K_b \neg H_a \bot) & \text{$b$ does not know whether  $a$ is correct}
\end{array}
\]
\item \emph{Self-correction under uncertainty of who self-corrects.}
Recall that the hope update formula for self-correction of $a$ generally has form  $\neg H_a \bot \lor (\phi \land K_a H_a \bot)$. Instead of two agents, as in Example~\ref{example.selfcorrect}, now consider any number $n = |\A|$ of agents.  Of course, the difference with \cref{example.selfcorrect}  only kicks in if $n \geq 3$. 

We can encode that an arbitrary agent self-corrects, while the remaining agents are uncertain which agent this is, by a hope update model consisting of $n$ events $e_1,\dots,e_n$ where event~$e_i$ represents that agent $i$ self-corrects. We now set $\vartheta_i(e_i) \ce \neg H_i\bot \lor (\psi_i \land K_i H_i \bot)$ for each~$i \in \A$ (where $\psi_i$~is some optional constraint for agent~$i$ to self-correct) and $\vartheta_j(e_i) \ce \neg H_j\bot$ for all~$j \neq i$. We let each agent be unable to distinguish among any events wherein it does not self-correct: 
\[
e_i \KK^U_j e_k\quad\Longleftrightarrow\quad  i = j = k  \text{ or }j \notin\{i,k\} .
\] 
Thus, if an agent considers it possible that multiple agents know that they are incorrect, then after this update such an agent would generally not know whether somebody self-corrected and, if so, who it was.
\item \emph{Self-correction under uncertainty of the source of state recovery.}
An alternative generalization of \cref{example.selfcorrect} is that it remains public that a given agent~$a$ self-corrects but  there is uncertainty over the agent from whom agent $a$ can get its state recovery information, which can be encoded via formulas $\phi_i$ in $a$'s hope update formulas $\neg H_{a}\bot \lor (\phi_i \land K_{a} H_a \bot)$, for $i \in \A$ with $i \neq a$ (we assume that $a$ does not get the recovery information from itself). Among these $\phi_i$ the recovering agent $a$ non-deterministically chooses one. This is implemented in a hope update model of size~$n-1$, with events $e_i$ for all $i \neq a$ such that $\vartheta_a(e_i) \ce \neg H_a\bot \lor (\phi_i \land K_a H_a \bot)$ and \mbox{$\vartheta_{j}(e_i)=\neg H_j \bot$} for all $j \ne a$, and such that $\KK^U_a$ is the identity on this domain of events (as~$a$~knows what choice it makes between the $\phi_i$'s), whereas  for $i \neq a$, relation~$\KK^U_i$ is the universal relation on this domain (any other agent remains uncertain among all these alternatives).\looseness=-1

\end{enumerate}
\end{example}

\subsection{Axiomatization}

\begin{definition}[Axiomatization $\axKH^\priv$]
\label{axiomatization2}
$\axKH^\priv$ extends $\axKH$ with axioms 
\[\begin{array}{lll@{\qquad}lll}
{[U,e]}p &\eq& p 
& {[U,e]}\neg \varphi &\eq& \neg [U,e]\varphi 
\\ 
{[U,e]}(\varphi \et \psi) &\eq& [U,e]\varphi \et [U,e]\psi 
& 
{[U,e]}[U',e']\varphi &\eq& \bigl[(U;U'),(e,e')\bigr]\varphi 
\\[.5ex] 
{[U,e]} K_i \varphi &\eq& \Et\limits_{e \KK^U_i f} K_i [U,f]\varphi 
&
{[U,e]} H_i \varphi &\eq& \left(\vartheta_i(e) \imp \Et\limits_{e \KK^{U}_i f} K_i \bigl(\vartheta_i(f) \imp [U,f]\phi\bigr)\right) 
\end{array}\]
where $\phi, \psi \in \lanhopekn^{\textit{priv}}$, $(U', e')$ is a pointed hope update model, $p \in \Prop$, $i \in \mathcal{A}$, and $U = (E,\vartheta,\KK^U)$ is a hope update model with $e,f \in E$.
\end{definition}

\begin{theorem}[Soundness]
\label{theorem:soundness2}
For all $\phi \in \lanhopekn^\priv$, $\axKH^{\priv} \vdash \phi$ implies $\KK\HH \models \phi$.
\end{theorem}
\begin{proof}
As in \cref{theorem:soundness}, it is sufficient to show the validity of the new axioms. Additionally, the proofs for first three axioms for atomic propositions, negation, and conjunction are similar to those for the analogous axioms of $\axKH^{\pub}$ and of action model logic, so are omitted here.
For the remaining three axioms, consider arbitrary Kripke model $M=(W,\pi,\KK,\HH) \in \KK\HH$ with $w \in W$, as well as hope update models $U = (E,\vartheta,\KK^U)$ and $U' = (E',\vartheta',\KK^{U'})$ with $e,f \in E$ and $e' \in E'$. Let $M \times U =  (W^\times,\pi^\times,\KK^\times, \HH^\times)$ according to \cref{def.semanticssemi} and $(U ; U')=(E'',\vartheta'',\KK^{U ; U'})$ according to \cref{def:comp}. To avoid unnecessary clutter, in this proof we use single parentheses instead of double ones, e.g., writing $\HH^\times_i (w,e)$ instead of $\HH^\times_i \bigl((w,e)\bigr)$.
\begin{itemize}
\item  
Axiom ${[U,e]} K_i \varphi \eq \Et_{e \KK^U_i f} K_i [U,f]\varphi $ is valid because
\\
$M, w \models [U,e] K_i \varphi$%
\quad if{f}\quad
$M \times U,  (w,e) \models K_{i} \phi$%
\quad if{f}\\ 
$\bigl(\forall (v,f) \in \KK^\times_i (w,e)\bigr)\,\, M \times U,  (v,f) \models \phi$%
\quad  if{f}\quad
$\bigl(\forall (v,f) \in \KK^\times_i (w,e)\bigr)\,\,M,v   \models [U,f]\phi$%
\quad if{f}\\ 
$(\forall v \in W)(\forall f\in E) \bigl(v \in \KK_i (w)\,\&\, e \KK^{U}_i f \,\, \Longrightarrow\,\, M, v \models  [U,f]\phi\bigr)$%
\quad if{f}\\ 
$(\forall f\in E)\Bigl(e \KK^{U}_i f\, \Longrightarrow\, 
\bigl(\forall v \in  \KK_i (w)\bigr)\,\, M, v \models  [U,f]\phi\Bigr)$%
\quad if{f}\quad\\ 
$(\forall f\in E)\bigl(e \KK^{U}_i f\, \Longrightarrow\, 
M, w \models K_i[U,f]\phi\bigr)$%
\quad if{f}\quad 
$M, w \models \Et_{e \KK^{U}_i f} K_i [U,f]\phi$
\item  
Axiom ${[U,e]} H_i \varphi \eq \left(\vartheta_i(e) \imp \Et_{e \KK^{U}_i f} K_i \bigl(\vartheta_i(f) \imp [U,f]\phi\bigr)\right)$ is valid because
\\
$M, w \models [U,e] H_i \varphi$%
\quad if{f}\quad
$M \times U,  (w,e) \models H_{i} \phi$%
\quad if{f}\\ 
$\bigl(\forall (v,f) \in \HH^\times_i (w,e)\bigr)\,\, M \times U,  (v,f) \models \phi$%
\quad  if{f}\\ 
$\bigl(\forall (v,f) \in \KK^\times_i (w,e)\bigr)\,\,  \bigl(M,w \models \vartheta_i(e) \,\&\, M,v \models \vartheta_i(f) \quad\Longrightarrow\quad M \times U,  (v,f) \models \phi\bigr)$%
\quad if{f}\\ 
$M, w \models \vartheta_i(e)\,\,\Longrightarrow\,\,\bigl(\forall (v,f) \in \KK^\times_i (w,e)\bigr)\,\,\bigl(M, v \models \vartheta_i(f)\quad\Longrightarrow\quad M,v   \models [U,f]\phi\bigr)$%
\quad if{f}\\ 
$M, w \models \vartheta_i(e)\,\,\Longrightarrow\,\,\bigl(\forall (v,f) \in \KK^\times_i (w,e)\bigr)\,\, M, v \models \vartheta_i(f)\rightarrow [U,f]\phi$%
\quad if{f}\\ 
$M, w \models \vartheta_i(e)\,\,\Longrightarrow\,\,(\forall v \in W)(\forall f\in E) \bigl(v \in \KK_i (w)\,\&\, e \KK^{U}_i f \,\, \Longrightarrow\,\, M, v \models \vartheta_i(f)\rightarrow [U,f]\phi\bigr)$%
\quad if{f}\\ 
$M, w \models \vartheta_i(e)\,\,\Longrightarrow\,\,(\forall f\in E)\Bigl(e \KK^{U}_i f\, \Longrightarrow\, 
\bigl(\forall v \in  \KK_i (w)\bigr)\,\, M, v \models \vartheta_i(f)\rightarrow [U,f]\phi\Bigr)$%
\quad if{f}\\ 
$M, w \models \vartheta_i(e)\,\,\Longrightarrow\,\,(\forall f\in E)\Bigl(e \KK^{U}_i f\, \Longrightarrow\, 
M, w \models K_i\bigl(\vartheta_i(f)\rightarrow [U,f]\phi\bigr)\Bigr)$%
\quad if{f}\\ 
$M, w \models \vartheta_i(e)\,\,\Longrightarrow\,\,
M, w \models \Et_{e \KK^{U}_i f} K_i\bigl(\vartheta_i(f)\rightarrow [U,f]\phi\bigr)$%
\quad if{f}\\ 
$M, w \models \vartheta_i(e) \rightarrow \Et_{e \KK^{U}_i f} K_i\bigl(\vartheta_i(f) \rightarrow [U,f]\varphi\bigr)$.

\item
To show the validity of axiom ${[U,e]}[U',e']\varphi \eq \bigl[(U;U'),(e,e')\bigr]\varphi $, we first  show that models $(M \times U) \times U'$ and $M \times (U ; U')$ are isomorphic.  It is easy to see that models  $(M \times U) \times U' = (W^{\times\!\times},\pi^{\times\!\times},\KK^{\times\!\times},\HH^{\times\!\times})$ and $M \times (U ; U')=(W^{;},\pi^{;},\KK^{;},\HH^{;})$ where
\begin{itemize}
\item 
$W^{\times\!\times}=(W \times E) \times E'$;
\item
$W^{;} = W \times (E \times E')$;
\item 
$\pi^{\times\!\times}(p) = \bigl\{\bigl((w,e),e'\bigr) \mid w \in \pi(p)\bigr\}$;
\item
$\pi^{;}(p) =\bigl\{\bigl(w,(e,e')\bigr) \mid w \in \pi(p)\bigr\}$;
\item 
$\bigl((w,e),e'\bigr) \KK_i^{\times\!\times} \bigl((v,f),f'\bigr)$%
\quad if{f} \quad 
$w \KK_i v$, and $e \KK^U_i f$, and $e' \KK^{U'}_i f'$;
\item 
$\bigl(w,(e,e')\bigr) \KK_i^{;} \bigl(v,(f,f')\bigr)$
\quad if{f} \quad 
$w \KK_i v$, and $e \KK^U_i f$, and $e' \KK^{U'}_i f'$;
\item 
$\bigl((w,e),e'\bigr) \HH_i^{\times\!\times} \bigl((v,f),f'\bigr)$%
\quad if{f}\\\strut\hfill 
$w \KK_i v$, and $e \KK^U_i f$, and $e' \KK^{U'}_i f'$, and $M\times U,(w,e)\models \vartheta'_i(e')$, and $M\times U,(v,f)\models \vartheta'_i(f')$;
\item 
$\bigl(w,(e,e')\bigr) \HH_i^{;} \bigl(v,(f,f')\bigr)$%
\quad if{f}\\\strut\hfill 
$w \KK_i v$, and $e \KK^U_i f$, and $e' \KK^{U'}_i f'$, and $M,w\models [U,e]\vartheta'_i(e')$, and $M,v\models [U,f]\vartheta'_i(f')$.
\end{itemize}
It is immediate that 
\begin{gather*}
\bigl((w,e),e'\bigr) \in \pi^{\times\!\times}(p) 
\quad\Longleftrightarrow\quad 
\bigl(w,(e,e')\bigr) \in \pi^{;}(p),
\\
\bigl((w,e),e'\bigr) \KK_i^{\times\!\times} \bigl((v,f),f'\bigr)
\quad\Longleftrightarrow\quad 
\bigl(w,(e,e')\bigr) \KK_i^{;} \bigl(v,(f,f')\bigr),
\\
\bigl((w,e),e'\bigr) \HH_i^{\times\!\times} \bigl((v,f),f'\bigr)
\quad\Longleftrightarrow\quad 
\bigl(w,(e,e')\bigr) \HH_i^{;} \bigl(v,(f,f')\bigr).
\end{gather*}
Thus, function
 $f\colon W^{\times\!\times} \to W^; $ defined by
\[
f\colon \bigl((w,e),e'\bigr) \mapsto \bigl(w,(e,e')\bigr)
\] 
is an isomorphism between these models.
It remains to note that 
$M, w \models [U,e][U',e']\varphi$%
\quad if{f}\quad 
$M \times U, (w,e) \models [U',e']\varphi$%
\quad if{f}\quad 
$(M \times U) \times U', \bigl((w,e),e'\bigr) \models \varphi$. Due to isomorphism $f$, this is equivalent to 
$M \times (U ; U'), \bigl(w,(e,e')\bigr)\models \varphi$%
\quad if{f}\quad 
$M, w\models \bigl[(U;U'),(e,e')\bigr]\varphi$.\qed
\end{itemize}
\end{proof}

Similarly to the previous section, one can show that every formula in $\lanhopekn^\priv$ is provably equivalent to a formula in $\lanhopekn$. For that \cref{definition:weight} can be adapted
by defining complexity of  hope update models~$U= (E,\vartheta,\KK^U)$ to be
$c(U) \ce \max \bigl\{ c\bigl(\vartheta_i(e)) \mid i \in \A, e \in E \bigr\}$ and replacing the last clause in  \cref{definition:weight} with
\[
c\bigl([U,e]\phi\bigr)
\ce 
\bigl(c(U) + |E|\bigr) \cdot c(\phi).
\]
where $|E|$ is the number of actions in hope update model~$U$. It can be shown that \cref{unequal} also holds for all axioms from \cref{axiomatization2}. Based on these complexity-decreasing left-to-right reductions, a translation $t \colon \lanhopekn^\priv \to \lanhopekn$ can be defined by analogy with \cref{translation}. Essentially the same argument  as in \cref{termination} shows that this translation is a terminating rewrite procedure. Thus:
\begin{proposition}[Termination]
\label{termination2}
For all $\phi \in \lanhopekn^\priv$, $t(\phi) \in \lanhopekn$.
\end{proposition}

The same argument as in \cref{lem:trans} and \cref{th:pub_complete} yields
\begin{lemma}[Equiexpressivity]
\label{lem:private_express}
Language $\lanhopekn^\priv$ is equiexpressive with $\lanhopekn$, i.e., for all formulas $\phi \in \lanhopekn^\priv$, $\axKH^{\priv} \vdash \phi \leftrightarrow t(\phi)$.
\end{lemma}

\begin{theorem}[Soundness and completeness]
\label{th:private_comp}
For all $\phi \in \lanhopekn^\priv$,
\[
\axKH^{\priv} \vdash \phi \qquad \Longleftrightarrow\qquad \KK\HH \models \phi.
\]
\end{theorem}
Finally, as in \cref{cor:nec}, necessitation for private hope update is an admissible inference rule in~$\axKH^{\priv}$. In other words, if $\axKH^{\priv} \vdash \phi$, then $\axKH^{\priv} \vdash [U,e]\phi$.

\section{Factual Change} \label{sec.factual}

In this section, we provide a way to add factual change to our model updates.  
This is going along well-trodden paths in dynamic epistemic logic \cite{hvdetal.aamas:2005,jfaketal.lcc:2006,hvdetal.world:2008}.

\subsection{Syntax, Semantics, and Axiomatization}

\begin{definition}[Hope update model with factual change] \label{def.defdefdef}
To obtain a \emph{hope update model with factual change} $U = (E,\vartheta,\sigma,\KK^U)$ from a hope update model $(E,\vartheta,\KK^U)$ for a language $\lang$ we add  parameter $\sigma: E \imp (\Prop \imp \lang)$. We require that each $\sigma(e)$ be only finitely different from the identity function, i.e., that the set $\{p \in \Prop \mid \sigma(e)(p) \ne p\}$ be finite for each $e \in E$.
\end{definition}
The finitary requirement is needed in order to keep the language well-defined. 

\begin{definition}[Language $\lanhopekn^f$] \label{def.yy}
Language $\lanhopekn^f$ is defined by the grammar that looks like the one in \cref{def.y} except that $(U,e)$ in the clause $[U,e] \phi$ here is a pointed hope update model with factual change for the language $\lanhopekn^f$. 
\end{definition}
As in the previous section, \cref{def.yy} is given by mutual recursion and from here on all hope update models are hope update models with factual change for  language $\lanhopekn^f$.

\begin{definition}[Semantics]\label{def.semanticsfact}
Let $U = (E,\vartheta,\sigma,\KK^U)$ be a hope update model, $M = (W,\pi,\KK, \HH) \in \KK\HH$, $w \in W$, and $e \in E$. Then, the only new clause compared in \cref{def.semanticssemi} is replaced by a different update mechanism
\[
M,w \models [U,e]\phi \quad \text{if{f}}\quad M \otimes U, (w,e) \models \phi,
\] with $M \otimes U = (W^\times,\pi^\otimes,\KK^\times, \HH^\times)$ with the same $W^\times$, $\KK^\times$, and $\HH^\times$ as in \cref{def.semanticssemi} and such that:
\[
\begin{array}{l@{\quad}l@{\quad}l}
W^\times  &\ce &  W \times E;
\\
(w,e) \in \pi^\otimes(p)  &\text{if{f}} &M,w \models \sigma(e)(p); \\
(w,e) \KK^\times_i (v,f) &\text{if{f}} & w \KK_i v \text{ and } e \KK^U_i f;\\
\smallskip
(w,e) \HH^\times_i (v,f) &\text{if{f}} &(w,e) \KK^\times_i (v,f), \text{ and } M,w \models \vartheta_i(e), \text{ and } M,v \models \vartheta_i(f).
\end{array}
\]
\end{definition}
The only difference between \cref{def.semanticssemi,def.semanticsfact} is that the clause for the valuation~$\pi^{\times}$ of the former is: $(w,e) \in \pi^\times(p)$ if{f} $w \in \pi(p)$. 
In other words, there the valuation of facts does not change, and the valuation in the world $w$ is carried forward to that in the updated worlds $(w,e)$. Since class~$\KK\HH$ has no restrictions on valuations, it follows from \cref{proposition:semi-privateupdatestayinkh} that $M \otimes U \in \KK\HH$ whenever $M  \in \KK\HH$.

To follow the familiar pattern of reduction axioms from $\axKH^\pub$ and $\axKH^\priv$, we first need to adapt the composition operation. For the composition $U;U' = (E'', \vartheta'', \sigma'', \KK^{U;U'})$ of hope update models $U = (E, \vartheta, \sigma, \KK^U)$ and $U' = (E', \vartheta', \sigma', \KK^{U'})$ with factual change,  the new parameter~$\sigma''$ needs to be added to \cref{def:comp} (cf.~\cite{hvdetal.world:2008}): for any $(e,e'), (f,f') \in  E \times E'$, 
\[\begin{array}{l@{\quad}l@{\quad}l}
E'' & \ce & E \times E' \\
\smallskip
\vartheta''_i\bigl((e,e')\bigr) & \ce & [U,e] \vartheta_i'(e')\\
\smallskip
\sigma''\bigl((e,e')\bigr)(p) & \ce & 
\begin{cases}
[U,e]\sigma'(e')(p) & \text{if $\sigma'(e')(p) \ne p$},\\
\sigma(e)(p) & \text{if $\sigma'(e')(p) = p$ but $\sigma(e)(p) \ne p$},\\
p & \text{if $\sigma'(e')(p) = \sigma(e)(p) = p$}
\end{cases}
\\
\smallskip
(e,e') \KK^{U ; U'}_i (f,f') & \text{if{f}} & e \KK^{U}_i f \text{ and } e' \KK^{U'}_i f'
\end{array}\]

With this upgrade to the composition of hope update models, the only required change to the axiom system $\axKH^\priv$ from \cref{axiomatization2} is replacing the first equivalence  with $[U,e]p \eq \sigma(e)(p)$:
\begin{definition}[Axiomatization $\axKH^{f}$]
\label{axiomatization3}
$\axKH^{f}$ extends $\axKH$ with axioms 
\[\begin{array}{lll@{\qquad}lll}
[U,e]p &\eq& \sigma(e)(p)
& {[U,e]}\neg \varphi &\eq& \neg [U,e]\varphi 
\\ 
{[U,e]}(\varphi \et \psi) &\eq& [U,e]\varphi \et [U,e]\psi 
& 
{[U,e]}[U',e']\varphi &\eq& \bigl[(U;U'),(e,e')\bigr]\varphi 
\\[.5ex] 
{[U,e]} K_i \varphi &\eq& \Et\limits_{e \KK^U_i f} K_i [U,f]\varphi 
&
{[U,e]} H_i \varphi &\eq& \left(\vartheta_i(e) \imp \Et\limits_{e \KK^{U}_i f} K_i \bigl(\vartheta_i(f) \imp [U,f]\phi\bigr)\right) 
\end{array}\]
where $\phi, \psi \in \lanhopekn^{f}$, $(U', e')$ is a pointed hope update model with factual change, $p \in \Prop$, $i \in \mathcal{A}$,  $U = (E,\vartheta,\sigma,\KK^{U\mathstrut})$ is a hope update model with factual change, and $e,f \in E$.
\end{definition}

\begin{theorem}[Soundness]
For all $\phi \in \lanhopekn^f$, $\axKH^f \vdash \phi$ implies $\KK\HH \models \phi$.
\end{theorem}
\begin{proof}
For most of the new axioms the proof of \cref{theorem:soundness2} transfers to this case verbatim. To show the validity of ${[U,e]}[U',e']\varphi \eq \bigl[(U;U'),(e,e')\bigr]\varphi$, the proof follows along the same lines by showing that $(M \otimes U)  \otimes U'$ is isomorphic to $M \otimes (U;U')$, with the argument for the valuations replaced with that from \cite[Prop.~2.9]{hvdetal.world:2008}.  Finally, it is easy to see that 
$\KK\HH \models [U,e]p \eq \sigma(e)(p)$,
as  $M, w \models [U,e]p$ if{f} $M \otimes U, (w,e) \models p$ if{f} $(w,e) \in \pi^\otimes(p)$ if{f} $M,w \models \sigma(e)(p)$. 
\qed
\end{proof}
In itself it is quite remarkable that the required changes are fairly minimal, given the enormously enhanced flexibility in specifying distributed system behavior. From this point the techniques used for $\axKH^\priv$ apply with barely a change to factual change. The same arguments as for \cref{lem:private_express} (for an appropriately modified complexity measure) and \cref{th:private_comp} yield the analogous statements for $\lanhopekn^f$, once again with the admissibility of necessitation rule as a corollary:

\begin{lemma}[Equiexpressivity]
Language $\lanhopekn^f$ is equiexpressive with $\lanhopekn$.
\end{lemma}

\begin{theorem}[Soundness and completeness]
For all $\phi \in \lanhopekn^f$, 
\[
\axKH^f \vdash \phi
\qquad \Longleftrightarrow \qquad
\KK\HH \models \phi.
\]
\end{theorem}

\subsection{Applications}

The importance of adding factual change to our framework comes from the fact that, in practical protocols implementing
FDIR mechanisms, agents usually take decisions based on what they recorded in their local states. 
We demonstrate the essentials of combined hope updates and state recovery in \cref{example.reset}, which combines the variant of
self-correction introduced in \cref{example.selfcorrect} with state recovery needs
that would arise in the alternating bit protocol~\cite{halpernzuck:1992}.

\begin{example}[Private self-correction with state recovery] 
\label{example.reset}
The alternating bit protocol (ABP) for transmitting an arbitrarily generated stream of consecutive data packets $d_1,d_2,\dots$ from a sender to a receiver over an unreliable communication channel uses messages that additionally contain a sequence number consisting of 1 bit only. The latter switches from one message to the next, by alternating atomic propositions $q_s$ and $q_r$ containing the next sequence number to be used for the next message generated by the sender resp. receiver side of the channel. 
In addition, the sender maintains atomic proposition $p_s$, using the difference between the two to kick-start sending of the next packet. The receiver would not need this second bit in the absence of faults. We use $p_r$ for
self-correction, however, in the sense that we assume that it provides a reliable backup
for $q_r$. In the fault-free case, it will be maintained such that the invariant $p_r \neq q_r$ holds.
If the receiver becomes faulty, we assume that its FDIR unit supplies the value
$q_r$ needs to be recovered to as~$\neg p_r$.

Let us describe several consecutive steps of how the ABP should operate in more detail with agent $s$ being the sender and agent $r$ the receiver. Suppose agents have the values $(q_s,q_r)=(0,0)$ and $(p_s,p_r)=(1,1)$ of their local variables when the sending of data packet $d_n$ begins. The sending of $d_n$ and the next packet $d_{n+1}$ happens in six phases (three per packet)~(\cite{halpernzuck:1992}), where we describe actions of each agent in term of its local variables: 
\begin{enumerate}[(i)]
\item\label{phase_one} Since $q_s\neq p_s$ (here $0 \ne 1$),  sender~$s$ sets $p_s \ce q_s=0$ and generates a message \mbox{$(d_n,p_s)=(d_n,0)$} to be repeatedly sent to $r$. \\\strut\hfill Local values in this phase are  $(q_s,q_r)=(0,0)$ and $(p_s,p_r)=(0,1)$.
\item When receiver~$r$ receives $(d_n,q_r)=(d_n,0)$, it records $d_n$, generates a message\linebreak $(ack,q_r)=(ack,0)$ to be repeatedly sent back to $s$, and switches to the next sequence number $q_r\ce1-q_r=1$, also updating the backup $p_r \ce 1-p_r=0$.  \\\strut\hfill Local values in this phase are  $(q_s,q_r)=(0,1)$ and $(p_s,p_r)=(0,0)$.

\item\label{phase_three} When  sender $s$ receives $(ack,p_s)=(ack,0)$,  sender switches to the next sequence number $q_s\ce 1- p_s=1$ and next data packet $d_{n+1}$. \\\strut\hfill Local values in this phase are  $(q_s,q_r)=(1,1)$ and $(p_s,p_r)=(0,0)$.
\item Since $q_s\neq p_s$ (here $1 \ne 0$),  sender~$s$ sets $p_s \ce q_s=1$ and generates a message\linebreak \mbox{$(d_{n+1},p_s)=(d_{n+1},1)$} to be repeatedly sent to $r$. \\\strut\hfill 
Local values in this phase are  $(q_s,q_r)=(1,1)$ and $(p_s,p_r)=(1,0)$.
\item When receiver~$r$ receives $(d_{n+1},q_r)=(d_{n+1},1)$, it records $d_{n+1}$, generates a message\linebreak $(ack,q_r)=(ack,1)$ to be repeatedly sent back to $s$, and switches to the next sequence number $q_r\ce1-q_r=0$, also updating the backup $p_r \ce 1-p_r=1$.
 \\\strut\hfill Local values in this phase are  $(q_s,q_r)=(1,0)$ and $(p_s,p_r)=(1,1)$.

\item When  sender $s$ receives $(ack,p_s)=(ack,1)$, sender switches to the next sequence number $q_s\ce 1- p_s=0$. \hfill Local values in this phase are  $(q_s,q_r)=(0,0)$ and $(p_s,p_r)=(1,1)$.
\end{enumerate}
At this point, all local variables have returned to their values when $s$ had started sending packet~$d_n$, and the cycle repeats again and again.  Thus, during a correct run of the ABP, values of $(q_s,q_r)$ continuously cycle through $(0,0)$, $(0,1)$, $(1,1)$, $(1,0)$, $(0,0)$. Note also that, $p_r \ne q_r$ throughout any correct run of the protocol, enabling to retrieve a correct value of $q_r$ from backup $p_r$. By contrast, $p_s = q_s$ can happen, creating an asymmetry between sender and receiver.

If a transient fault would flip the value of either $q_s$ or $q_r$, the ABP can deadlock and, therefore, would require correction. For instance, if $q_s$ flips from $1$ to $0$ at the end of phase~\eqref{phase_three}, the condition~$p_s \ne q_s$ for the start of sending $d_{n+1}$ would never be fulfilled.

Due to the above mentioned invariant $p_r \ne q_r$, 
the need for a correction of receiver  (in case $q_r$~has accidentally flipped)
can be conveniently determined by checking whether $p_r=q_r$, while the correction itself can be performed by just setting $q_r\ce 1- p_r$.

To model this self-correction in our logic, we treat boolean variables $p_r$, $q_r$, $p_s$, and $q_s$ as atomic propositions so that $p_r = q_r$ becomes $p_r \eq q_r$ and $q_r\ce 1- p_r$ looks like $q_r \ce \neg p_r$.
Accordingly, 
we model agent $r$ successfully self-correcting and recovering its state based on the condition $p_r \leftrightarrow q_r$.
At the same time, $s$ is uncertain whether $r$~has corrected itself (event $scr_{p_r = q_r}$) or not (event $noscr$). Again writing $\vartheta(e)$ as $\bigl((\vartheta_a(e),\vartheta_b(e)\bigr)$, this is encoded in the hope update model $U\ce(E, \vartheta, \sigma, \KK^{U})$, where:

\[\begin{array}{lcl@{\quad}|@{\quad}lcl}
E & \ce & \{scr_{p_r = q_r}, noscr\} & \sigma(scr_{p_r = q_r})(q_r)&\ce& \lnot p_r  \\
\vartheta (scr_{p_r = q_r}) &\ce &\bigl(\neg H_s \bot, \neg H_r \bot \lor (p_r \leftrightarrow q_r)\bigr) & \KK^{U}_{s} &\ce& E \times E \\
\vartheta (noscr) &\ce& (\neg H_s \bot,\neg H_r \bot)  & \KK^{U}_{r} &\ce& \{(e,e) \mid e \in E\}
\end{array}\]
Note that $H_r \bot$ is  equivalent to $p_r \leftrightarrow q_r$, making $H_r \bot$ locally detectable by $r$ and resulting in $\vartheta (scr_{p_r = q_r}) = (\neg H_s \bot, \top)$. In other words, agent $r$ is guaranteed to become correct whenever this update is applied. All atoms for  $noscr$ and all atoms other than $q_r$ for  $scr_{p_r = q_r}$ remain unchanged by $\sigma$.
Coding the atoms in each state as $p_sq_s.p_rq_r$, the resulting update is:
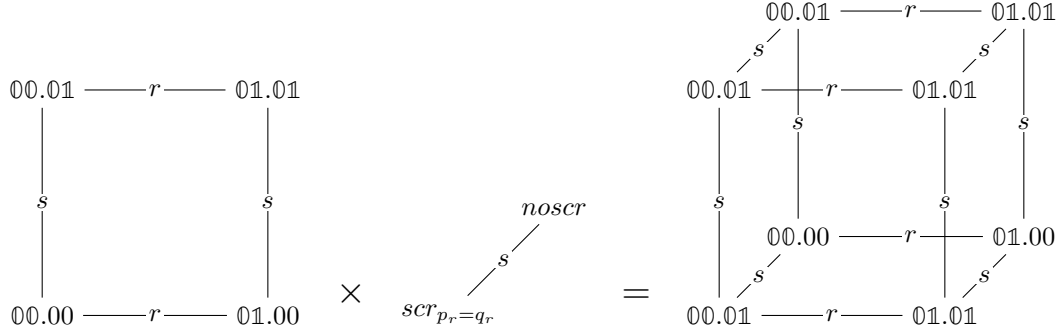
\[
\begin{tikzpicture}[z=0.35cm]
\node (000) at (0,0,0) {$\ourpmb{00}.00$};
\node (010) at (0,3,0) {$\ourpmb{00}.\ourpmb{01}$};
\node (100) at (3,0,0) {$\ourpmb{01}.00$};
\node (110) at (3,3,0) {$\ourpmb{01}.\ourpmb{01}$};
\draw (000) -- node[fill=white,inner sep=1pt] {$r$} (100);
\draw (000) -- node[fill=white,inner sep=1pt] {$s$} (010);
\draw (010) -- node[fill=white,inner sep=1pt] {$r$} (110);
\draw (100) -- node[fill=white,inner sep=1pt] {$s$} (110);
\end{tikzpicture}
\quad \raisebox{.4cm}{\Large $\times$} \quad 
\begin{tikzpicture}[z=0.35cm]
\node (000) at (0,0,0) {$scr_{p_r = q_r}$};
\node (010) at (0,0,4) {$noscr$};
\draw (000) -- node[fill=white,inner sep=1pt] {$s$} (010);
\end{tikzpicture}
\quad \raisebox{.4cm}{\Large $=$} \quad 
\begin{tikzpicture}[z=0.35cm]
\node (000) at (0,0,0) {$\ourpmb{00}.\ourpmb{01}$};
\node (001) at (0,0,3) {$\ourpmb{00}.00$};
\node (010) at (0,3,0) {$\ourpmb{00}.\ourpmb{01}$};
\node (011) at (0,3,3) {$\ourpmb{00}.\ourpmb{01}$};
\node (100) at (3,0,0) {$\ourpmb{01}.\ourpmb{01}$};
\node (101) at (3,0,3) {$\ourpmb{01}.00$};
\node (110) at (3,3,0) {$\ourpmb{01}.\ourpmb{01}$};
\node (111) at (3,3,3) {$\ourpmb{01}.\ourpmb{01}$};
\draw (000) -- node[fill=white,inner sep=1pt] {$r$} (100);
\draw (001) -- node[fill=white,inner sep=1pt] {$r$} (101);
\draw (011) -- node[fill=white,inner sep=1pt] {$r$} (111);

\draw (000) -- node[fill=white,inner sep=1pt] {$s$} (010);
\draw (001) -- node[fill=white,inner sep=1pt] {$s$} (011);
\draw (101) -- node[fill=white,inner sep=1pt] {$s$} (111);

\draw (000) -- node[fill=white,inner sep=1pt] {$s$} (001);
\draw (010) -- node[fill=white,inner sep=1pt] {$s$} (011);
\draw (100) -- node[fill=white,inner sep=1pt] {$s$} (101);
\draw (110) -- node[fill=white,inner sep=1pt] {$s$} (111);
\draw (010) -- node[fill=white,inner sep=1pt] {$r$} (110);
\draw (100) -- node[fill=white,inner sep=1pt] {$s$} (110);
\end{tikzpicture}
\]
The only change happens in 
global states $\ourpmb{00}.00$ and $\ourpmb{01}.00$ where $p_r \leftrightarrow q_r$ causes the hope update and $q_r$ is set to be the opposite of~$p_r$. 
After the update, we get:
\[\begin{array}{l@{\quad}l}
M, {00}.00 \models [U,scr_{p_r = q_r}] (\neg H_{r} \bot \land  K_r q_r) & \text{$r$ has corrected herself and learned the right value of $q_r$}\\
M, {00}.00 \models [U,scr_{p_r = q_r}]  K_{r} \neg H_{r} \bot & \text{$r$ is now sure she is correct}\\
M, {00}.00 \models [U,scr_{p_r = q_r}] (\neg K_{r} q_s \land \neg K_r \neg q_s) & \text{$r$ remains unsure regarding $q_s$} \\
M, {00}.00 \models [U,scr_{p_r = q_r}] \M_{s}  H_{r} \bot & \text{$s$ considers possible that $r$ is faulty}
\end{array}\]
\end{example}

 \section{Conclusions and Further Research}
 \label{sec:conclusions}

 We gave various dynamic epistemic semantics for the modeling and analysis of byzantine fault-tolerant multi-agent systems, expanding a known logic containing knowledge and hope modalities. We provided complete axiomatizations for our logics and applied them to fault-detection, isolation, and recovery (FDIR) in distributed computing. For future research we envision alternative dynamic epistemic update mechanisms, as well as  embedding our logic into the (temporal epistemic) runs-and-systems approach.

\paragraph*{Acknowledgments.} We are grateful for multiple fruitful discussions with and  enthusiastic support from Giorgio Cignarale, Stephan Felber, Rojo Randriano\-men\-tsoa, Hugo Rinc\'on Galeana, and Thomas Schl\"ogl.

\bibliographystyle{abbrvurl}
\newcommand{\benthem}[1]{} \providecommand{\noopsort}[1]{}

\end{document}